\RequirePackage{amsmath}
\documentclass[letterpaper,11pt]{article}
\usepackage{amsthm}
\usepackage{amsmath,amssymb, mathtools}
\usepackage{framed,caption,color,url}
\usepackage{standalone}
\usepackage{times}
\usepackage{fullpage}
\usepackage{algorithm2e,algorithmicx,algpseudocode,tikz,wrapfig}
\usepackage{graphicx}               
\usepackage{cite} 
\usepackage{nicefrac} 
\usepackage{paralist} 
\usepackage{boxedminipage} 
\usepackage[pdfstartview=FitH,colorlinks,linkcolor=blue,filecolor=blue,citecolor=blue,urlcolor=blue,pagebackref=true]{hyperref}
\usepackage{upgreek}

\usepackage{centernot}

\newcommand{\fot}{\bff}
\newcommand{\fto}{\bfg}

\newcommand{\fhat}[2]{\ifthenelse{\equal{#2}{}}{\hat{f}(#1)}{\ifthenelse{\equal{#2}{0}}{\hat{f}(\emptyset)}{\hat{f}(#1_{\leq #2})}}}
\newcommand{\ftild}[2]{\ifthenelse{\equal{#2}{}}{\tilde{f}(#1)}{\ifthenelse{\equal{#2}{0}}{\tilde{f}(\emptyset)}{\tilde{f}(#1_{\leq #2})}}}
\newcommand{\ftildstar}[2]{\ifthenelse{\equal{#2}{}}{\tilde{f^*}(#1)}{\ifthenelse{\equal{#2}{0}}{\tilde{f^*}(\emptyset)}{\tilde{f^*}(#1_{\leq #2})}}}
\newcommand{\ghat}[2]{\ifthenelse{\equal{#2}{}}{\hat{g}(#1)}{\ifthenelse{\equal{#2}{0}}{\hat{g}(\emptyset)}{\hat{g}(#1_{\leq #2})}}}
\newcommand{\mfix}[2]{\ifthenelse{\equal{#2}{}}{m(#1)}{\ifthenelse{\equal{#2}{0}}{m(\emptyset)}{m(#1_{\leq #2})}}}

\newcommand{\fapp}[2]{\ifthenelse{\equal{#2}{}}{\tilde{f}(#1)}{\ifthenelse{\equal{#2}{0}}{\tilde{f}(\emptyset)}{\tilde{f}(#1_{\leq #2})}}}

\newcommand{\Otilde}{\wt{O}}

\newcommand{\metric}{{{\mathsf{d}}}}

\usepackage[utf8]{inputenc}
\newcommand{\Levy}{Lévy\xspace}
\newcommand{\Lovasz}{Lovász\xspace}

\newcommand{\parag}[1]{\medskip \noindent{\bf #1}}

\newcommand{\HD}{\mathsf{HD}}

\newcommand{\msr}{{\bm{\upmu}}}

\newcommand{\OffTam}{\mathsf{Off}\Tam}

\newcommand{\Tam}{\mathsf{Tam}}
\newcommand{\TamAb}{\mathsf{TamAb}}
\newcommand{\AppTam}{\mathsf{AppTam}}
\newcommand{\AppTamAb}{\mathsf{AppTamAb}}

\newcommand{\yDist}{\mathbf{y}}

\newcommand{\tDist}{\mathbf{t}}

\newcommand{\uDist}{\mathbf{u}}
\newcommand{\vDist}{\mathbf{v}}
\newcommand{\wDist}{\mathbf{w}}

\newcommand{\yDistVec}{\ol{\yDist}}

\newcommand{\tDistVec}{\ol{\tDist}}

\newcommand{\uDistVec}{\ol{\uDist}}
\newcommand{\vDistVec}{\ol{\vDist}}
\newcommand{\wDistVec}{\ol{\wDist}}

\newcommand{\xVec}{\ol{x}}
\newcommand{\yVec}{\ol{y}}
\newcommand{\zVec}{\ol{z}}

\newcommand{\uVec}{\ol{u}}
\newcommand{\vVec}{\ol{v}}
\newcommand{\wVec}{\ol{w}}

\newcommand{\alVec}{{\ol{\alpha}}}

\newcommand{\alk}{{\bfK_\alVec}}
\newcommand{\alHD}{\HD_{\alVec}}

\renewcommand{\th}{^\mathrm{th}}
\newcommand{\jointSamp}[2]{\langle #1 \,\, \| \, {#2}\rangle}

\newcommand{\VP}{\mathrm{ValPrf}}

\newcommand{\pfix}[2]{ {#1}_{\leq #2}}

\newcommand{\aSF}{\mathsf{a}}
\newcommand{\gSF}{\mathsf{g}}
\newcommand{\hSF}{\mathsf{h}}

\newcommand{\avr}[2]{\ifthenelse{\equal{#2}{}}{\aSF({#1})}{\ifthenelse{\equal{#2}{0}}{\aSF(\emptyset)}{\aSF({#1}_{\leq #2})}}}

\newcommand{\avrMax}[2]{\ifthenelse{\equal{#2}{}}{\aSF^*({#1})}{\ifthenelse{\equal{#2}{0}}{\aSF^*(\emptyset)}{\aSF^*({#1}_{\leq #2})}}}

\newcommand{\avrApp}[2]{\ifthenelse{\equal{#2}{}}{\tilde{\aSF}({#1})}{\ifthenelse{\equal{#2}{0}}{\tilde{\aSF}(\emptyset)}{\tilde{\aSF}({#1}_{\leq #2})}}}

\newcommand{\avrAppMax}[2]{\ifthenelse{\equal{#2}{}}{\tilde{\aSF}^*({#1})}{\ifthenelse{\equal{#2}{0}}{\tilde{\aSF}^*(\emptyset)}{\tilde{\aSF}^*({#1}_{\leq #2})}}}

\newcommand{\ArgMax}[2]{\ifthenelse{\equal{#2}{}}{\hSF({#1})}{\ifthenelse{\equal{#2}{0}}{\hSF(\emptyset)}{\hSF({#1}_{\leq #2})}}}

\newcommand{\AppArgMax}[2]{\ifthenelse{\equal{#2}{}}{\tilde{\hSF}({#1})}{\ifthenelse{\equal{#2}{0}}{\tilde{\hSF}(\emptyset)}{\tilde{\hSF}({#1}_{\leq #2})}}}

\newcommand{\gain}[2]{\ifthenelse{\equal{#2}{}}{\gSF(#1)}{\gSF(#1_{\leq #2})}}
\newcommand{\gainMax}[2]{\ifthenelse{\equal{#2}{}}{\gSF^*(#1)}{\gSF^*(#1_{\leq #2})}}
\newcommand{\gainApp}[2]{\ifthenelse{\equal{#2}{}}{\tilde{\gSF}(#1)}{\tilde{\gSF}(#1_{\leq #2})}}
\newcommand{\gainAppMax}[2]{\ifthenelse{\equal{#2}{}}{\tilde{\gSF}^*(#1)}{\tilde{\gSF}^*(#1_{\leq #2})}}

\newcommand{\pr}[2][]{\Pr_{\ifthenelse{\isempty{#1}}{}{{#1}}}\left[{#2}\right]}

\newcommand{\sm}{\setminus}

\newcommand{\e}{\mathrm{e}}

\usepackage{graphicx}

\newcommand{\remove}[1]{}

\newcommand{\ol}{\overline}
\newcommand{\wt}[1]{\widetilde{#1}}

\newcommand{\es}{\varnothing} 
\newcommand{\se}{\subseteq}

\newcommand{\set}[1]{\left\{ #1 \right\}}

\newcommand{\bits}{\{0,1\}}

\newcommand{\norm}[1]{\left\lVert#1\right\rVert}
\newcommand{\To}{\mapsto}

\newcommand{\R}{{\mathbb R}}
\newcommand{\N}{{\mathbb N}}

\newcommand{\Adv}{\mathsf{A}}

\newcommand{\cS}{{\mathcal S}}
\newcommand{\cT}{{\mathcal T}}

\newcommand{\cX}{{\mathcal X}}
\newcommand{\cY}{{\mathcal Y}}

\newcommand{\bfK}{\mathbf{K}}

\newcommand{\bff}{\mathbf{f}}
\newcommand{\bfg}{\mathbf{g}}
\newcommand{\bfh}{\mathbf{h}}

\newcommand{\bfk}{\mathbf{k}}

\newcommand{\bft}{\mathbf{t}}

\newcommand{\bfx}{\mathbf{x}}
\newcommand{\bfy}{\mathbf{y}}

\usepackage{upgreek}
\usepackage{bm}

\newcommand{\eps}{\varepsilon}
\newcommand{\vphi}{\varphi}

\newcommand{\poly}{\operatorname{poly}}
\newcommand{\polylog}{\operatorname{polylog}}

\newcommand{\Exp}{\operatorname*{\mathbb{E}}}
\newcommand{\Ex}{\Exp}

\newcommand{\Var}{\operatorname*{\mathbb{V}}}

\newcommand{\negl}{\operatorname{negl}}
\newcommand{\Supp}{\operatorname{Supp}}

\newcommand{\argmax}{\operatorname*{argmax}}

\newtheorem{theorem}{Theorem}[section]
\theoremstyle{plain}
\newtheorem{claim}[theorem]{Claim}

\newtheorem{lemma}[theorem]{Lemma}
\newtheorem{corollary}[theorem]{Corollary}

\theoremstyle{definition}

\newtheorem{definition}[theorem]{Definition}
\newtheorem{construction}[theorem]{Construction}

\theoremstyle{definition}
\newtheorem{remark}[theorem]{Remark}

\newcommand{\sdotfill}{\textcolor[rgb]{0.8,0.8,0.8}{\dotfill}} 

\makeatletter
\def\th@protocol{%
    \normalfont 
    \setbeamercolor{block title example}{bg=orange,fg=white}
    \setbeamercolor{block body example}{bg=orange!20,fg=black}
    \def\inserttheoremblockenv{exampleblock}
  }
\makeatother
\theoremstyle{protocol}

\newtheorem{proto}[theorem]{Protocol}
\newtheorem{protoc}[theorem]{Protocol}

\newcommand{\namedref}[2]{#1~\ref{#2}}

\newcommand{\torestate}[3]{%
\expandafter \def \csname BBRESTATE #2 \endcsname{#3}
\theoremstyle{plain}
\newtheorem{BBRESTATETHMNUM#2}[theorem]{#1}
\begin{BBRESTATETHMNUM#2}\label{#2}\csname BBRESTATE #2 \endcsname   \end{BBRESTATETHMNUM#2}
\newtheorem*{BBRESTATETHMNONNUM#2}{\namedref{#1}{#2}}
}

\newcommand{\restate}[1]{\begin{BBRESTATETHMNONNUM#1}[Restated] \csname BBRESTATE #1 \endcsname
\end{BBRESTATETHMNONNUM#1}}

\title{Computational Concentration of Measure: \\ Optimal Bounds, Reductions, and More
}
\author{Omid Etesami\thanks{Institute for Research in Fundamental Sciences (IPM).} \and Saeed Mahloujifar\thanks{University of Virginia. Supported by  University of Virginia's SEAS Research Innovation Award.} \and Mohammad Mahmoody\thanks{University of Virginia. Supported by NSF CAREER award CCF-1350939  and UVa's SEAS Research Innovation Award.}}

\newcommand{\Mnote}[1]{}
\newcommand{\Onote}[1]{}
\newcommand{\Snote}[1]{}

\begin{document}
\maketitle

\begin{abstract}
Product measures of dimension $n$ are known  to be ``concentrated'' under Hamming distance. More precisely, for any set  $\cS$ in the product space of probability $\Pr[\cS] \geq \eps$, 
a random point in the space, with probability $1-\delta$,
has a neighbor in $\cS$ that is different from the original point in only $O(\sqrt{n\cdot\ln(\nicefrac{1}{\eps \delta})})$ coordinates (and this is optimal).
In this work, we obtain the tight \emph{computational} (algorithmic) version of this result, showing how given a random point and access to an $\cS$-membership query oracle, we can find such a close point of Hamming distance  $O(\sqrt{n\cdot\ln(\nicefrac{1}{\eps \delta})})$ in time $\poly(n,1/\eps,1/\delta)$. This resolves an open question of \cite{mahloujifar19-ALT} who proved a weaker result (that works only for $\eps\gg 1/\sqrt n$). As corollaries, we obtain 
polynomial-time poisoning and (in certain settings) evasion attacks against learning algorithms when the original vulnerabilities have any cryptographically non-negligible probability.

    We call our algorithm MUCIO (short for ``MUltiplicative Conditional Influence Optimizer") since proceeding through the coordinates of the product space, it decides to change each coordinate of the given point based on a multiplicative version of the influence of a variable, where the influence is computed conditioned on the value of all previously updated coordinates.
MUCIO is an online algorithm in that it decides on the $i$'th coordinate of the output given only the first $i$ coordinates of the input.
It also does not make any convexity assumption about the set $\cS$.

Motivated by obtaining algorithmic variants of measure concentration in other metric probability spaces, we define a new notion of algorithmic reduction between computational concentration of measure in different probability metric spaces. This notion, whose definition has some subtlety, requires \emph{two} (inverse) algorithmic mappings one of which is an algorithmic Lipschitz  mapping and the other one is an algorithmic coupling connecting the two distributions.
As an application, we apply this notion of reduction to obtain computational concentration of measure for high-dimensional Gaussian distributions under the $\ell_1$ distance.

We further prove several extensions to the results above as follows. (1) Generalizing in another dimension, our computational concentration result is also true when the Hamming distance is weighted.  (2) As measure concentration is usually proved for concentration around mean, we show how to use our results above to obtain algorithmic concentration for that setting as well. In particular, we prove a computational variant of McDiarmid's inequality, when properly defined. (3)  Our result generalizes to discrete random processes (instead of just product distributions), and this generalization leads to new tampering algorithms for collective coin tossing protocols.
(4) Finally, we prove exponential lower bounds on the average running time of \emph{non-adaptive} query algorithms for proving computational concentration for the case of product spaces. 
Perhaps surprisingly, such lower bound shows any efficient algorithm must query about $\cS$-membership of points that are \emph{not} close to the original point even though we are only interested in finding a close point in $\cS$. 


\end{abstract}


\thispagestyle{empty}
\tableofcontents
\thispagestyle{empty}
\clearpage
\setcounter{page}{1}

\section{Introduction}
Let $(\cX,\metric,\msr)$ be a metric probability space in which $\metric$ is a metric over $\cX$, and  $\msr$ is a probability measure over $\cX$. The concentration of measure phenomenon \cite{ledoux2001concentration,milman1986asymptotic} states that many  natural metric probability spaces of high dimension are concentrated in the following sense. Any  set $\cS \se \cX$ of ``not too small'' probability $\msr(\cS) \geq \eps$ is ``close'' (according to $\metric$) to ``almost all''  points ( $1-\delta$ measure according to $\msr$). 

A well-studied class of concentrated spaces is the set of  product spaces in which the measure $\msr = \msr_1 \times \dots \msr_n$  is a product measure of dimension $n$, and the metric $\metric$ is Hamming distance of dimension $n$; namely, $\HD(\uVec,\vVec) = |\set{i \colon u_i \neq v_i}|$ for vectors $\uVec=(u_1,\dots,u_n),\vVec=(v_1,\dots,v_n)$. More specifically, it is known, e.g., by results implicit in \cite{amir1980unconditional,milman1986asymptotic} and explicit in \cite{mcdiarmid1989method,talagrand1995concentration}, and weaker versions known as blowing-up lemma proved in \cite{AhlswedeGacsKorner1976, Margulis1974, Marton1986}, that any such metric probability space is a so-called Normal \Levy family \cite{levy1951problemes,alon1985lambda1}. Namely, for any $\cS$ of probability $\msr(\cS) \geq \eps$, 
at least $1-\delta$ fraction  of the points  (under the product measure $\msr$) are $O(\sqrt{n\cdot \ln(\nicefrac{1}{\eps \delta})})$-close in Hamming distance to $\cS$. 
Previous proofs of measure concentration, and in particular those proofs for product spaces are \emph{information theoretic}, and only show the \emph{existence} of a ``close'' such point $\yVec \in \cS$ to most of $\xVec \gets \msr$ sampled according to $\msr$. Naive sampling of points around $\xVec$ will likely \emph{not} fall into $\cS$ (see Section~\ref{sec:lower-bound}).


Motivated by finding polynomial-time attacks on the ``robustness'' of machine learning algorithms, recently Mahloujifar and Mahmoody \cite{mahloujifar19-ALT} studied a \emph{computational} variant of the measure concentration in which the mapping from a given point $\xVec \gets \msr$ to its close neighbor $\yVec \in \cS$ is supposed to be computed by an efficient polynomial-time algorithm $A^{\cS,\msr}(\xVec)=\yVec$ that has oracle access to test membership in $\cS$ and a sampling oracle from the measure $\msr$.\footnote{In case of product measure, oracle access to a sampler from $\msr = \msr_1 \times \dots \msr_n$ is equivalent to having such samplers for all $\msr_i$.}  It was shown in \cite{mahloujifar19-ALT} that if $\cS$ is large enough, then the measure computationally concentrates around $\cS$. In particular, it was shown that if $\Pr[\cS]\geq 1/\polylog(n)$, then $A^{\cS,\msr}(\xVec)$ finds $\yVec$ with Hamming distance $\Otilde(\sqrt n)$ from $\xVec$, and instead if $\cS$ is at least $\Pr[\cS] \geq \omega(1/\sqrt n)$, then $A$ finds $\yVec$ with Hamming distance $o(n)$. Consequently, it was left open to prove computational concentration of measure around any smaller sets of ``non-negligible'' ${1}/{\poly(n)}$ probability, e.g., of measure $1/n$.

\subsection{Our Results} In this work, we  resolve the open question about the computational concentration of measure in product spaces under Hamming distance and prove (tight up to constant) computational concentration for all range of initial probabilities $\Pr[\cS]$ for the target set $\cS$. Namely, we prove the following result matching what information theoretic concentration of product spaces guarantees up to a constant factor, while the mapping is done algorithmically. As we deal with algorithms, without loss of generality, we focus on discrete distributions.\footnote{Note that even seemingly non-discrete distributions like Gaussian, when used as input to efficient algorithms, are necessarily rounded to limited precision and thus end up being discrete.}

\begin{theorem}[Main result] \label{thm:mainInf}
There is an algorithm $A_{\eps,\delta}^{\cS,\msr}(\cdot)$ called MUCIO (short for ``MUltiplicative Conditional Influence Optimizer") that given access to a membership oracle for any set $\cS$ and a sampling oracle from any product measure $\msr$ of dimension $n$, it  achieves the following. If $\Pr[\cS] \geq \eps$, given $\eps$ and $\delta$, the  algorithm $A_{\eps,\delta}^{\cS,\msr}(\cdot)$ runs in time $\poly(\nicefrac{n}{\eps \delta})$, and with probability $\ge 1-\delta$ given a random point $\xVec \gets \msr$, it maps $\xVec$ to a point $\yVec \in \cS$ of bounded Hamming distance $\HD(\xVec,\yVec) \leq O(\sqrt{n\cdot \ln(\nicefrac{1}{\eps \delta})})$. 

\end{theorem}

See Theorem~\ref{thm:main} for a more general version of  Theorem~\ref{thm:mainInf}.

For the special case that $\eps,\delta ={1}/{\poly(n)}$ (implying $\cS$ has a non-negligible measure)
the algorithm MUCIO of Theorem~\ref{thm:mainInf} achieves its goal in $\poly(n)$ time, while it changes only $\Otilde(\sqrt{n})$ of the coordinates.

Our work can be seen as another example of works in computer science that  make previously existential proofs algorithmic. A good example of a similar successful effort is the active line of work started from \cite{moser2009constructive,moser2010constructive} that presented algorithmic proofs of \Lovasz's local lemma, leading to  algorithms that efficiently find objects that previously where only shown to exist using \Lovasz's local lemma. 
The work of \cite{impagliazzo2010constructive} also approaches measure concentration from an algorithmic perspective, but their goal is to algorithmically find witness for \emph{lack} of concentration.

\subsubsection{Extensions}
 In this work we also prove several extensions to our main result in different directions expanding a direct study of computational concentration as an independent direction.

\parag{Extension to random processes and coin-tossing attacks.} We prove a more general result than Theorem~\ref{thm:mainInf} in which the perturbed object   is a random process. Namely, suppose $\wDistVec \equiv (\wDist_1,\dots,\wDist_n)$ is a discrete (non-product) random process in which, given the history of blocks $w_1,\dots,w_{i-1}$, the $i\th$ block $w_i$ is sampled from its corresponding random variable $(\wDist_i \mid w_1,\dots,w_{i-1})$. 
Suppose  $\Pr_{\wVec \gets \wDist}[\wVec \in \cS] \geq \eps$ for an arbitrary set $\cS$. A natural question is: how much can an adversary increase the probability of falling into $\cS$, if it is allowed to partially tamper with the online process of sampling $w_1,\dots,w_n$ up to $K<n$ times? In other words, the adversary has a limited budget of $K$, and in the $i\th$ step, it can use one of its budget, and in exchange it gets to override the originally (honestly) sampled value $w_i \gets (\wDist_i \mid w_1,\dots,w_{i-1})$ by a new value. Note that if the adversary does a tampering, the changed value will \emph{substitute} $w_i$ and will  affect the way the future blocks of the random process are sampled, e.g., in the next sampling of $ w_{i+1} \gets (\wDist_{i+1} \mid w_1,\dots,w_{i})$.

Our generalized version of Theorem~\ref{thm:mainInf} (stated in Theorem~\ref{thm:main}) shows that in the above setting of tampering with random processes, an adversary with budget $O(\sqrt{n\cdot \ln(\nicefrac{1}{\eps \delta}}))$ can indeed change the distribution of the random process and make the resulting tampered sequence end up in $\cS$ with probability at least $1-\delta$, while the adversary also runs in  time $\poly(\nicefrac{n}{\eps \delta})$. Previously, \cite{mahloujifar19-ALT} also showed a similar less tight  result for random processes, but their result was limited to the setting that $\cS$ is sufficiently large $\Pr[\cS] \geq \omega(1/\sqrt n)$.

The variant of Theorem~\ref{thm:mainInf} for  random processes allows us to attack cryptographic coin-tossing protocols \cite{ben1989collective,cleve1993martingales,maji2010computational,berman2014coin,haitner2014coin,kalai2018lower} in which $n$ parties $P_1,\dots,P_n$ each send a single message during a total of $n$ rounds, and the full transcript $M=(m_1,\dots,m_n)$ determines a bit $b$. The goal of an attacker is to corrupt up to $K$ of the parties and bias the bit $b$ towards its favor. Our results show that even if the original bit $b$ had a small probability of being $1$, $\Pr_{\text{no-attack}}[b=1]\geq \eps=1/\poly(n)$, then a $\poly(n)$-time attacker who can corrupt up to $\Otilde(\sqrt{n})$ parties and change their messages can bias the output bit $b$ all the way up to make it $\Pr_{\text{attack}}[b=1]\geq 1-1/\poly(n)$. The corruption model here was first introduced by Goldwasser, Kalai and Park \cite{goldwasser2015adaptively} and is called  \emph{strong} adaptive corruption, because the adversary has the option to first see the message $m_i$ before deciding to corrupt (or not corrupt) $P_i$ to change its message $m_i$ (or not).  \footnote{If each message $m_i$ is a bit, it turns out that our attack can be modified to an attack that is not strong.}
 
\parag{Weighted Hamming distance.} In another extension to our Theorem~\ref{thm:mainInf} (see Theorem~\ref{thm:main}) we allow the Hamming distance to have different costs $\alpha_i$ when changing the $i\th$ coordinate for any vector $\alVec = (\alpha_1,\dots,\alpha_n)$ of $\ell_2$ norm $\sum_i \alpha_i^2=n$. 
In Talagrand's inequality \cite{talagrand1995concentration}, it is proved that even if $\alVec_{\xVec}$ can completely depend on the original point $\xVec$, we still can conclude that most points are ``close'' to any sufficiently large set $\cS$, when the distance from $\xVec$ to  $\cS$ is measured by the $\alVec_{\xVec}$-weighted Hamming distance. 
An algorithmic version of Talagrand's inequality, then, shall find a close point $\yVec \in \cS$ to $\xVec$ measured by $\alVec_{\xVec}$-weighed Hamming distance. Interestingly, our  proof  allows the coordinate $\alpha_i$ to completely depend on $(x_1,\dots,x_{i-1})$, but falls short of  proving an algorithmic version of Talagrand's inequality, if possible at all.

\parag{Reductions and other metric probability spaces.} 
Motivated by proving computational concentration of measure in  other metric probability spaces, as well as designing a machinery for this goal, we define a new model of \emph{algorithmic reductions} between computational concentration of measure in different  metric probability spaces. This notion, whose definition has some subtle algorithmic aspects, requires \emph{two} (inverse) polynomial-time mappings one of which is an algorithmic Lipschitz  mapping and the other one is an algorithmic coupling connecting the two distributions.
As an application, we apply this notion of reduction to obtain computational concentration of measure for high-dimensional Gaussian distributions under the $\ell_1$ distance. We prove this exemplary case by revisiting the proof of \cite{bobkov1997isoperimetric} 
who proved the \emph{information theoretic} reduction from the concentration of Gaussian distributions under the $\ell_1$ distance to that of Hamming cube. We show how the core ideas of \cite{bobkov1997isoperimetric} could be extended to obtain all the algorithmic components that are needed for a computational variant.  Although there are known results on concentration of Gaussian distribution $\ell_1$ in information theoretic regime, this is the first time (to the best of our knowledge) that a computational variant of concentration is proved for Gaussian spaces. We envision the same machinery can be applied to more information theoretic results for obtaining new computational variants; we leave doing so for future work. See Theorem~\ref{thm:reduction} for the formal statement.

\parag{Computational concentration around mean.}  As  measure  concentration is usually proved for concentration around mean of a function $f(\cdot)$ when the inputs come from certain distributions, we show how to use our main result of Theorem~\ref{thm:main} to obtain algorithmic concentration results for that setting as well. Namely, at a high level, we show that in certain settings (where concentration is known to follow from those settings) one can algorithmically find the right minimal perturbations to sampled points $\xVec$ so that the new perturbed point $\xVec'$ gives us the average of the concentrated function: $f(\xVec') \approx \Ex_{x \gets \msr}[f(\xVec)]$. Sometimes doing so is trivial (e.g., in case of Chernoff bound, when $f$ is simply the addition of i.i.d. sampled Boolean values, as one can greedily change Boolean variables to decrease their summation) but sometimes doing so is not straightforward.  In particular, we prove a computational variant of McDiarmid’s inequality. Namely, we show how to modify $\sqrt{n}$ coordinates of a vector $\xVec \gets \msr$ sampled from a product distribution $\msr$ of dimension $n$, such that $f(\xVec')$ gets arbitrary (i.e., $1/\poly(n)$) close to the average $\mu = \Ex_{\xVec \gets \msr}[f(\xVec)]$ for a function $f$ that is Lipschitz under Hamming distance. (Note that the Lipschitz  property is needed for the McDiarmid inequality as well). See Theorem~\ref{thm:AlgMcD} for the formal statement.

\parag{Lower bounds for simple methods.} We also prove exponential lower bounds on the query complexity of natural, yet restricted, classes of algorithms. Two such classes stand out: One is non-adaptive algorithms where the queries made do not depend on the answer of previous queries.
Another, natural class of algorithms are algorithms where all the queried points are at the distance where an acceptable final output may be at that distance. These lower bounds shed light on why perhaps some of the ideas behind our algorithm MUCIO are necessary, and that some simpler more straightforward algorithms are not as efficient.

\parag{Polynomial-time biasing attacks against extractors.} At a high level, our biasing attacks on random processes are also related  to impossibility results on extracting randomness from blockwise Santha-Vazirani sources \cite{SanthaV86,chor1988unbiased,BeigiEG17,ReingoldVW04,DodisOnPrSa04} and specifically the $p$-tampering and $p$-resetting attacks of \cite{bentov2016bitcoin,pTampTCC17,Mahloujifar2018:ALT}. In those attacks, an attacker might get to tamper each incoming block with an \emph{independent} probability $p$, and they can achieve a bias of magnitude $O(p)$ (in polynomial time). However, our attackers \emph{can choose} which blocks are the target of their tampering substitutions, but then achieve much stronger bias and almost fixing the output with much smaller    $o(n)$ number of tamperings.

\subsubsection{Polynomial-time Attacks on Robust Learning} 
Our results also have implications on (limits) of robust learning, which is also the focus of the work of \cite{mahloujifar19-ALT} where computational concentration of measure was also studied. We refer the reader to \cite{mahloujifar19-ALT} for a more in-depth treatment of the literature and settings for (attacks on) robust learning. For sake of completeness, below we describe the basic setting of such attacks and briefly discuss the implication of our computational concentration results to robust learning attacks.

 Suppose $L$ is a (deterministic) learning algorithm, taking as input a training set  $T$ consisting of $m$ iid sampled and labeled examples $T=\set{x_i,c(x_i)}_{i \in [m]}$ where $x_i \gets \msr$ for $i\in[m]$, and that $c(\cdot)$ is a concept function to be learned. Let $h=L(T)$ be the hypothesis that the learner produces based on the training set $T$. Main attacks against robustness of learners are studied during the training phase or the testing phase of a learning process. We describe the settings and previous work before explaining the implication of our new computational concentration results to those settings.

\paragraph{Poisoning attacks.} In a so-called data poisoning attack \cite{barreno2006can,biggio2012poisoning}, which is tightly related to Valiant's malicious noise model \cite{Valiant::DisjunctionsConjunctions,KearnsLi::Malicious,NastyNoise}, the adversary only tampers with the training phase and substitutes a small $p<1$ fraction of the examples in $T$ with other arbitrary examples, leading to a poisoned data set $\wt{T}$. The goal of the adversary, in general, is to make $L(\wt{T})$ produce a ``bad'' hypothesis $h \in \wt{H}$ (e.g., bad might mean having large risk or making a mistake on a particular test $x$ during the test time) where $\wt{H} \se H$ includes the set of all undesired hypothesis. It was  shown by \cite{mahloujifar2018curse} that the concentration of measure in product spaces (under Hamming distance) implies that in any such learning process, so long as $\Pr_T[L(T) \in \wt{H}] \geq \eps$, then an adversary $\Adv$ who changes $O(\sqrt{m\cdot \ln(\nicefrac{1}{\eps \delta})})$ of the training examples (and substitute them with still correctly labeled data) can increase the probability of producing a bad hypothesis in $\wt{H}$ to $\Pr_{\wt{T}\gets \Adv(T)}[L(\wt{T}) \in \wt{H}] \geq \delta$. It was left open whether such attack can be made polynomial time, or that perhaps computational intractability  can be leveraged to prevent such attacks. The work of \cite{mahloujifar19-ALT} showed how to make such attacks polynomial time, only for the setting where the probability of falling into $\wt{H}$ was already not too small, and in particular at least $\omega(1/\sqrt n)$, and also with looser bounds. Our Theorem~\ref{thm:mainInf} shows how to get such polynomial time evasion attacks for \emph{any non-negligible} probability $\eps \geq 1/\poly(n)$. In fact, as stated in Theorem~\ref{thm:mainInf}, our attack's complexity can gracefully adapt to  $\eps$. 

The previous attacks of \cite{mahloujifar2018curse,mahloujifar19-ALT} and our newer attacks of this work do not contradict recent exciting works in defending against poisoning attacks \cite{diakonikolas2016robust,lai2016agnostic,diakonikolas2018sever,prasad2018robust}, as those defenses either focus on learning parameters of distributions  or, even in the classification setting, they aim to bound the \emph{risk} of the  hypothesis, while we increase the probability of a bad Boolean property.\footnote{In fact, the challenge in those works is to obtain polynomial-time \emph{learners} in settings where inefficient robust methods were perhaps known in the robust statistics literature. The focus here, however, is to obtain polynomial-time \emph{attacks}.}

\paragraph{Evasion attacks.} In another active line of work, other types of attacks on learners are studied in which the adversary enters the game during the test time. In such so-called \emph{evasion} attacks \cite{biggio2014security,CarliniWagner,Szegedy:intriguing,GoodfellowEtAl:MakeMLRobust} that find ``adversarial examples'', the goal of the adversary is to perturb the test input $x$ into a ``close'' input $\wt{x}$ under some metric $\metric$ (perhaps because this small perturbation is imperceptible to humans) in such a way that this tampering makes the hypothesis $h$ make a mistake. In \cite{mahloujifar2018curse}, it was also shown that the concentration of measure can potentially lead to inherent evasion attacks, as long as the input metric probability space $(\cX,\metric,\msr)$ is  concentrated. This holds e.g., if the space is a Normal \Levy family  \cite{levy1951problemes,alon1985lambda1}.
The work of  \cite{mahloujifar19-ALT}  showed the existence of polynomial time evasion attacks with sublinear perturbations for classification tasks in which the input distribution is a $n$-dimensional product space (e.g., the uniform distribution over the hypercube) under Hamming distance. But their attacks could be applied  only when the original risk $\eps$ of the hypothesis $h$ is at least $\eps = \omega(1/\sqrt{n})$. However, standard PAC learners  (e.g., based on empirical risk minimization) can indeed achieve polynomially small risk $\eps = 1/\poly(m)$ where $m$ is the sample complexity. Our  Theorem~\ref{thm:mainInf} shows how to obtain polynomial-time attacks even in   the low-risk regime $\eps=1/\poly(n)$\footnote{Note that in the ``high dimensional'' setting where input dimension $n$ is huge, we can see the sample complexity $m$ bounded, which implies  $\eps\geq 1/\poly(m)$ if $\eps=1/\poly(n)$.} and perturb given samples  $x \gets \msr$ in $\Otilde(\sqrt n)$ coordinates and make the perturbed adversarial instance $\wt{x}$ misclassified with high probability. 

Our results of Section \ref{sec:reduc} show that one can also obtain polynomial time evasion attacks for classifiers whose inputs come from metric probability spaces that use metrics other than Hamming distance (e.g., Gaussian under $\ell_1$). Using the reductionist approach of Section \ref{sec:reduc} one can perhaps obtain more such results. Our attacks, however, do not rule out the possibility of robust classifiers for specific input distributions such as images or voice that is the subject of  recent intense research \cite{Szegedy:intriguing,CarliniWagner,DeepFool}, but they shed light on barriers for robustness in theoretically natural settings. 
See \cite{bubeck2018adversarial,degwekar2019computational} for more discussion on other possible barriers for robust learning.

\subsection{Technical Overview}
In this subsection, we describe the challenges and key ideas behind the proof of Theorem \ref{thm:mainInf} and some of its extensions. The extension for the concentration around mean (see Section \ref{sec:mean}) follows directly from the main result about concentration around noticeably large sets. Thus, we only focus on explaining ideas behind some other extensions to our result; namely how to obtain new results through carefully defined algorithmic reductions, and proving limits for the power of simple methods for proving computational concentration.

\parag{Setting.} (The reader might find the explanations for our notation at the beginning of Section \ref{sec:prelim} useful.) Suppose $\wDistVec \equiv (\wDist_1 \times \dots \times \wDist_n)$ is a random variable with a product distribution of dimension $n$.\footnote{As discussed above, our results extend to random processes as well, when formalized carefully, but for simplicity we focus on the interesting special case of product distributions.} Also, suppose the set $\cS \se \Supp(\wDist)$ is denoted by its characteristic function $f$, where $f(\wVec)=1$ iff $\wVec \in \cS$. The goal of the tampering algorithm $\Tam$ is to change as few as possible of the sampled blocks $\wVec = (w_1,\dots,w_n) \gets \wDistVec$ making the new vector $\vVec = (v_1,\dots,v_n)$ such that $f(\vVec) =1$ with high probability (over the both steps of sampling $\wVec$ and obtaining $\vVec$ from it). 

\medskip
\noindent Our starting point is the previous attack of \cite{mahloujifar19-ALT} that only proved computational concentration around large sets of measure $\Pr[\cS] \geq \omega(1/\sqrt n)$. The result of \cite{mahloujifar19-ALT}, in turn, was built upon techniques developed in the work of Komargodski, Raz, and Kalai~\cite{kalai2018lower} that presented an alternative simpler proof for a previously known result of Lichtenstein et al.~\cite{lichtenstein1989some}. Below, we first describe the high level ideas behind the approach of \cite{mahloujifar19-ALT,kalai2018lower}, and then we describe why that approach breaks down when $\cS$ gets smaller than $1/\sqrt n$, and thus fails to obtain the optimal information theoretic bounds for concentration. We then   describe our new techniques  to bypass this challenge and obtain  computational concentration with optimal bounds.

\paragraph{The high-level approach of \cite{mahloujifar19-ALT}.} As it turns out, the tampering algorithm of \cite{mahloujifar19-ALT}, as well as ours, do  not need to know $w_{i+1},\dots,w_n$ when deciding to change $w_i$ (into a different $v_i \neq w_i$) or leaving it as is (i.e., $w_i = v_i$). So, a useful notation to use is the partial expected values, capturing the chance of falling into $\cS$ (i.e., $f(\wVec)=1$) over the randomness of the remaining blocks.
$$\fhat{w_1,\dots,w_i}{} = \Ex_{(w_{i+1},\dots,w_n) \gets (\wDist_{i+1},\dots,\wDist_n)}[f(w_1,\dots,w_n)].$$

One obvious reason for working with $\fhat{\cdot}{}$ quantities is that they can be \emph{approximated} with arbitrary small $\pm 1/\poly(n)$ additive error. This can be done using the sampling oracle of the distribution of $\wDistVec \equiv \wDist_1 \times \dots \times \wDist_n$ and the oracle $f(\cdot)$ determining membership in $\cS$.

At a high level, the idea behind the attack of \cite{mahloujifar19-ALT} is to change $w_i$ only if this change allows us to increase $\fhat{\cdot}{}$ \emph{additively} by $+ \lambda$ for a parameter  $\lambda \approx 1/\sqrt n$. We first describe this attack, and then explain its challenges against obtaining optimal bounds and how we resolve them.

At a high level, the attack of \cite{mahloujifar19-ALT} tampers with the $i\th$ block (i.e., $w_i$), if just before or just after looking at $w_i$, we conclude that we can increase $\fhat{\cdot}{}$ by $ \lambda$.
\begin{construction}[Attack of \cite{mahloujifar19-ALT} oracle $\fhat{\cdot}{}$] \label{const:tampBlockMM}
 Suppose that we are given a prefix $\pfix{v}{i-1}$ that is finalized, and we are also given a candidate value $w_i$ for the $i$'th block (supposedly sampled from  $\wDist_i$) and we want to decide to keep it $v_i=w_i$ or change it $v_i \neq w_i$.
 Let $\lambda>0$ be a parameter of the attack to be chosen later, $v^*_i = \argmax_{y_i} \fhat{\pfix{v}{i-1},y_i}{}$ be the choice for $i$'th block that maximizes  $\fhat{\pfix{v}{i}}{}$, and let  $f^*=\fhat{\pfix{v}{i-1},v^*_i}{}$.
\begin{enumerate}
    \item (Case 1)  If $f^* \geq  \fhat{\pfix{v}{i-1}}{} \textcolor{red}{+\lambda}$, then output $v_i = v_i^*$ (regardless of $w_i$).
    \item (Case 2) Otherwise, if (by looking at $w_i$)  $\fhat{\pfix{v}{i-1},w_i}{}\leq  \fhat{\pfix{v}{i-1}}{} \textcolor{red}{-\lambda}$, then again output $v_i = v_i^*$.
    \item (Case 3) Otherwise, keep the value $w_i$ and output $v_i = w_i$.
\end{enumerate}
\end{construction}

\medskip
\noindent{\em Why this attack biases $f(\cdot)$ towards 1?} For simplicity, support $\Pr[\cS]=1/2$. Suppose we ``color'' different $i \in [n]$ depending on whether the tampering algorithm changes the $i\th$ block $w_i$ or not. If $v_i \neq w_i$ (tampering happened), color $i$  green, denoted by $i \in G$, and otherwise color $i$ red, denoted as $i \in R = [n] \sm G$. A simple yet extremely useful observation is that we can write $f(\vVec)$ as the sum of the \emph{changes} in $\fhat{v}{i}$ between consecutive $i$. Namely, if we let 
$\ghat{v}{i} = \fhat{v}{i}-\fhat{v}{i-1},$ then 
$$\fhat{v}{n} - \fhat{\es}{} = f(\vVec) - 1/2  = \sum_{i \in [n]} \ghat{v}{i}.$$  
This means that we have to study the affect of the green and red coordinates $i$ on how $\ghat{v}{i}$ behaves, because that will tell us how the final output bit is determined and distributed.

Construction \ref{const:tampBlockMM} is designed so that, whenever $i$ is green,  the partial expectation oracle $\fhat{v}{i}$ jumps up at least by $\lambda$ (i.e., $\ghat{v}{i} \geq \lambda$). So, the only damage (leading to falling outside $\cS$) could come from the red coordinates and how they change $\fhat{v}{i}$ downwards.  Let us now focus on the red coordinates $i\in R$. A simple inspection of Construction~\ref{const:tampBlockMM} shows that, the change in $\fhat{\cdot}{}$ captured by $\ghat{v}{i}$ is bounded in absolute value by $\lambda$, and that is the result of no-tampering for a block. Therefore, the summation of  $\ghat{v}{i}$ for red  coordinates $i$ would cancel out each other and, by the Azuma inequality, the probability that this summation is more than $1$ is at most $\exp(-1/(n \cdot \lambda^2)$. So, by choosing $\lambda \ll 1/\sqrt n$, the red coordinates cannot control the final bit, as with high probability this summation is less than one. This means that the outcome (whenever the  red coordinates do not fix the function) should be $1$, because the green coordinates only increase the $\fhat{\cdot}{}$ function.

\medskip
\noindent{\em Why the attack is efficient?} The efficiency of the attack follows form its effectiveness and the same argument described above. Namely, whenever the green coordinates are determining the output, it means that their total sum of of $\ghat{v}{i}$ is going from a specific number in $[0,+1]$ to $1$, and each time they jump up by at least $\lambda$, so they cannot be more than $n/\lambda$ green steps. Since we chose $\lambda = 1/\sqrt n$, the efficiency follows as well.

\paragraph{The challenge when $\Pr[\cS]=\Ex[\wVec]=\eps$ is too small.} The issue with the above approach is that whenever $\eps$ is too small (not around $1/2$) we need to pick $\lambda$ much smaller, so that the summation (i.e., the effect of the red coordinates does not make the function reach zero). Simple calculation shows that after the threshold $\eps \approx 1/\sqrt n$, the number of tampered (green) blocks would grow too much and eventually become \emph{more} than $n$. However, note that when we reach $n$ tamperings, it means  the attack's efficiency is meaningless.

\subsubsection{Our Approach (MUCIO: MUltiplicative Conditional Influence Optimizer)}
\paragraph{Main step 1:  tampering with \emph{multiplicatively} influential blocks.} Our first key idea is to judge whether a block is influential (and thus tamper it) based how much it can change the partial expectations in a \emph{multiplicative} way. (This is related to the notion of a \emph{log-likelihood ratio} in statistics and information theory.) Construction \ref{const:tampBlockSimple} below describes this simple change. However, as  we will see, doing this simple change will have big advantages as well as new challenges to be resolved. We will describe both the advantages and thew new challenges after the construction.

\begin{construction}[\emph{Multiplicative} online tampering using oracle $\fhat{\cdot}{}$] \label{const:tampBlockSimple}
The key difference between this attack and that of Construction \ref{const:tampBlockMM} is that here, in order to judge whether  tampering with the current $i\th$ block is worth it or not, we make the decision based on the \emph{multiplicative} gain (in how $\fhat{\cdot}{}$ changes) that this would give us. Namely, for the same setting of Construction \ref{const:tampBlockMM}, we do as follows.
\begin{enumerate}
    \item (Case 1)  If $f^* \geq \textcolor{red}{e^{\lambda} \cdot} \fhat{\pfix{v}{i-1}}{}$, then output $v_i = v_i^*$ (regardless of $w_i$).
    \item (Case 2) Otherwise, if $\fhat{\pfix{v}{i-1},w_i}{}\leq \textcolor{red}{e^{-\lambda} \cdot} \fhat{\pfix{v}{i-1}}{} $, then output $v_i = v_i^*$.
    \item (Case 3) Otherwise, keep the value $w_i$ and output $v_i = w_i$.
\end{enumerate}
\end{construction}
\medskip
\noindent{\em Main advantage: the output is fully biased.}
We first describe what  advantages the above change gives us, and then will discuss the remaining challenges. The key insight into why this is a better approach is that the tampering algorithm of Construction~\ref{const:tampBlockSimple} will \emph{always} lead to obtaining $f(\vVec)=1$ at the end (i.e., we always end up in $\cS$).
In order to see why this is a big difference, notice that if $\fhat{\pfix{w}{0}}{}=\eps$ is very small at the beginning and we tamper only based on additive differences (as is done in Construction \ref{const:tampBlockMM}), there is a possibility that we do \emph{not} tamper with the first block and end up at $\fhat{\pfix{w}{1}}{} = 0$. Such a problem does not happen when we decide on tampering based on multiplicative improvement, and every tiny chance of falling into $\cS$ is taken advantage of.

\medskip
\noindent{\em Only few tamperings happen.}
To analyze the number of tamperings that occur in the ``idealized'' attack of Construction \ref{const:tampBlockSimple} we keep track of $\ln\left({\fhat{v}{i}}/{\fhat{v}{i-1}}\right)$ as we go. We know that the output of function under the attack is always 1 which means:
$$\sum_{i=1}^n \ln \left(\frac{\fhat{v}{i}}{\fhat{v}{i-1}}\right) = \ln\left(\frac{\fhat{v}{n}}{\fhat{\es}{}}\right)= \ln\left(\frac{1}{\fhat{\es}{}}\right).$$
We again categorize the indices $i$ to red and green. Green set indicates the locations that the algorithm tampers with $w_i$ and red is the set of locations that tampering has not happened and $v_i=w_i$. For the red locations, we prove the following inequality that plays a key role in our analysis of the attack. One interpretation of this inequality is that we will now use $\ln (1/\fhat{v}{i-1})$ as a potential function that allows us keep track of, and  control, the number of tamperings.
$$\ln (1/\fhat{v}{i-1}) - \Ex_{v_i \gets \vDist[ \pfix{v}{i-1}]}[\ln(1/\fhat{v}{i})] \geq -\frac{\lambda^2}{2}.$$
This inequality follows from a Jensen Gap inequality on the natural logarithm function. For green locations, we increase $\ln(\fhat{v}{i})$ whenever we tamper by at least $\lambda$. Therefore, the overall effect of green locations on $\sum_{i=1}^n \ln({\fhat{v}{i}}/{\fhat{v}{i-1}})$ will be 
$$\lambda\cdot \Ex[\# \text{ of tampering}].$$
Combining these together we get the following:
$$\lambda\cdot \Ex[\# \text{ of tampering}] - \frac{n\cdot \lambda^2}{2} \leq \ln({1}/{\eps}).$$
Now we can optimize $\lambda$ to get the best inequality on the expected number of tampering. 

\medskip
\noindent{\em New challenge: obtaining good multiplicative approximations when $\fhat{\cdot}{}$ gets too small.}
Construction~\ref{const:tampBlockSimple} increases the average to 1 (i.e., we always end up in $\cS$) with small number of tampering. However we cannot implement that construction  in polynomial time. The problem is that it is hard to instantiate the oracle $\fhat{v}{i}$ polynomial time when the partial average gets close to $0$. To solve this issue, we add a step to the construction that makes the algorithm abort if the partial average goes below some threshold. 

\begin{construction}[Online tampering \emph{with abort} $\TamAb$ using  partial-expectations oracle] \label{const:tampBlockAbortSimple} This construction is identical to Construction \ref{const:tampBlockSimple}, except that whenever the  fixed prefix has a too small partial expectation $\fhat{\pfix{v}{i-1},w_i}{}$ (based on a new parameter $\tau$) we will abort. Also, in that case the tampering algorithm does not tamper with any future $v_i$ block either. Namely, we add the following ``Case 0'' to the previous steps:
\begin{itemize}
\item (Case 0) If $\fhat{\pfix{v}{i-1},w_i}{}\leq e^{-\tau}\cdot \eps$ abort. If had aborted before, do nothing. 
\end{itemize}
\end{construction}

\paragraph{Main step 2: showing that reaching low expectations is unlikely under the attack.}
To argue that the new construction does not hurt the performance of our algorithm by much, we show that the probability of getting a low $\fhat{v}{i}$ is small because of the way our algorithm works. The idea is that, our algorithm always guarantees that
$$-\lambda \leq \ln{({\fhat{v}{i}}/{\fhat{v}{i-1}})}\leq \lambda.$$
We also show that 
$$\Ex[\ln({{\fhat{v}{i}}/{\fhat{v}{i-1}}})]\geq  -\frac{\lambda^2}{2}.$$
This means that the sequence of $\ln\left({\frac{\fhat{v}{i}}{\fhat{v}{i-1}}}\right)$ forms an ``approximate'' sub-martingale difference sequence. We can use Azuma inequality to show that sum of this sequence will remain bigger than some small threshold, with high probability. After all, we can bound the probability of getting into Case 0 to be very small.

\subsubsection{More Computational Concentration Results through Algorithmic Reductions} 
Here we explain a technical overview of our generic reduction technique. Let $S_1 = (\cX_1,\metric_1,\msr_1)$ and $S_2= (\cX_2,\metric_2,\msr_2)$ be two metric probability spaces. In addition, assume we already know some level of computational concentration proved for $S_2$, and that we want to prove (some level of) computational concentration  for $S_1$ through a reduction. In Section~\ref{sec:reduc}, we formalize a generic framework to prove such reductions. The main ingredients of such algorithmic reduction are two polynomial time mappings $\fot\colon \cX_1\to\cX_2$ and $\fto\colon\cX_2\to\cX_1$ with 3 properties. The first property (roughly speaking) requires that $\fot(\msr_1) \approx\msr_2$ and $\fto(\msr_2) \approx \msr_1$. This property guarantees that if we sample a point from one space and use the mapping and go to the other space, we get a distribution close to the probability measure of the second space. (This can be interpreted as an algorithmic coupling.) The second property requires that the mapping $\fto$ is Lipschitz. The third property requires that $\fto(\fot(x))$ is close to $x$. The idea behind why such reduction (as a collection of these mappings) work is as follows. We are given a point $x_1$ in $S_1$ and we want to find a close $x_2$ such that $x_2$ falls inside a subset $\cS$.  To do that we first map $x_1$ to a point $x'_1$ in $S_2$ using $\fot$. We know that $S_2$ is computationally concentrated and we can efficiently find a close $x'_2$ such that $x'_2$ falls into an specific subset $\cS'$. Then we use $\fto$ to go back to a point $x_2$ in $S_1$. The second and third properties together guarantee that $x_1$ and $x_2$ are close, because $x'_1$ and $x'_2$ are close. At the same time, the first condition guarantees that $x_2$ will hit $\cS$ if we select $\cS'$ in a careful way. See Theorem \ref{thm:reduction} for more details.

We use this general framework to prove computational concentration bounds for Gaussian spaces under $\ell_1$ norm. We reduce the computational concentration of Gaussian distribution under $\ell_1$ to the computational concentration of the Boolean Hamming cube. For this goal, we show how to build two mappings $\fot$ and $\fto$ from an $n$-dimensional Gaussian space to a $n^2$-dimensional Hamming cube and vice versa, following the footsteps of a reduction by \cite{bobkov1997isoperimetric} who proved an information theoretic variant of this result. Here we show that the algorithmic ingredients that are necessary, in addition to the ideas already in \cite{bobkov1997isoperimetric}, could indeed be obtained. The main idea behind this mappings is the fact that the number of $1$'s in a sample from $n$-dimensional hamming cube approximately forms a Gaussian distribution centered around $\frac{n}{2}$. Therefore, we can map each dimension of the Gauss space to a $n$-dimensional hamming cube and vice versa.  Here we observe that we can use the same idea and build the mappings in a way that achieves the three properties mentioned above.  See Section \ref{sec:reduc} for more details.

\subsubsection{Lower Bounds for Simple Methods}
To prove exponential lower bounds on the query complexity of too-simple algorithms, we consider the half-space $\cS$ in the Hamming cube consisting of those points with below-average Hamming weight. 

A uniformly random point $\xVec$ in the cube, with high probability has Hamming distance $\Omega(\sqrt{n})$ from the set $\cS$. Now, if for such a point $\xVec$, we hope to find a close point in $\cS$ simply by sampling uniformly at random among points close to $x$, we fail except with exponentially small probability. For only random points with distance  $n^{1-o(1)}$ have a significant chance of changing the weight of point $\xVec$ by $\Omega(\sqrt{n})$, whereas the information-theoretic bound says there exists a point of distance $O(\sqrt{n})$ that changes the weight by $\Omega(\sqrt{n})$.

To achieve lower bounds for more general classes of algorithms, we use a random half-space instead of a fixed half-space. 
This gives us exponential lower bounds for non-adaptive attacks as well as attacks that query about $\cS$-membership of points outside a ball of size $d = O(\sqrt{n\cdot \ln(\nicefrac{1}{\eps \delta})})$ even when we are interested in finding a point in the intersection of $\cS$ and this ball.
Notice that MUCIO avoids this last restriction
by surveying the influence of the first coordinate on the totality of points, while it ends up changing only a small fraction of the coordinates.

\section{Preliminaries} \label{sec:prelim}
\paragraph{General notation.} We  use calligraphic letters (e.g., $\cX$) for sets. By default, all distributions and random variables in this work are discrete. We use bold letters (e.g., $\wDist$) to denote random variables that return a  sample from a corresponding discrete distribution.
By $w \gets \wDist$ we denote sampling $w$  from the  random variable $\wDist$. By $\Supp(\wDist)$ we denote the support set of $\uDist$.  For an event $\cS \se \Supp(\wDist)$, the probability function of $\wDist$ for $\cS$ is denoted as $\Pr[\wDist \in \cS] = \Pr_{w \gets \wDist}[w \in \cS]$.
For a randomized algorithm $R(\cdot)$, by $y \gets R(x)$ we denote the randomized execution of $R$ on input $x$ outputting $y$. 
By $\uDist \equiv \vDist$ we denote that the random variables $\uDist$ and $\vDist$ have the same marginal distributions. 
Unless stated otherwise, we denote vectors by using a bar  over a variable.  By $\wDistVec \equiv (\wDist_1,\wDist_2,\dots,\wDist_n)$ we refer to a sequence of $n$ \emph{jointly sampled} random variables. 
For a vector $\wVec=(w_1\dots w_n)$, we use $\pfix{w}{i}$ to denote the prefix $(w_1,\dots, w_i)$, and we use the same notation $\pfix{\wDist}{i}$ for jointly distributed random variables. 
%
 For a jointly distributed random variables $(\uDist,\vDist)$, by $(\uDist \mid v)$ we denote the conditional distribution $(\uDist \mid \vDist = v)$. 
For a random variable  $\uDist$, by $T^{\uDist}(\cdot)$ we denote an oracle-aided algorithm $T^{(\cdot)}(\cdot)$ that can query fresh sample from $\uDist$. By $\uDist \times \vDist$ we refer to the product distribution in which $\uDist$ and $\vDist$ are sampled independently.
For a real-valued random variable $\bfx$, by $\Ex[\bfx]$ we refer to the expected value of $\bfx$, and by $\Var[\bfx]$ we denote its variance. 


\paragraph{Notation on random processes and online samplers.} 
Let $\wDistVec \equiv (\wDist_1, \dots , \wDist_n)$ be a sequence of jointly distributed random variables. We can interpret the distribution of $\wDistVec$ as a random process in which the $i\th$ block $w_i$ is sampled from the marginal distribution $(\wDist_i \mid \pfix{w}{i-1})$. For simplicity, we use notation $\wDist[\pfix{w}{i-1}]  \equiv (\wDist_i \mid \pfix{w}{i-1})$ to refer to this marginal conditional distribution. (Note that $i$ is dropped from the distribution's name, relying on the input $\pfix{w}{i-1}$ that uniquely determines $i$.)  We can interpret $\pfix{w}{i-1}$ as a ``node'' in a tree of depth $i$, and the sampling $w_i \gets \wDist[\pfix{w}{i-1}]$ can be seen as the process of sampling the next child according to the distribution of $\wDist[\pfix{w}{i-1}]$. Alternatively, describing the distributions of the random variables $\wDist[\pfix{w}{i-1}]$ defines the distribution of $\wDistVec$. For random variable  $\wDistVec \equiv (\wDist_1,\dots,\wDist_n)$ we sometimes refer to the random variable $\wDist[\pfix{w}{i-1}]$ as the \emph{online} sampler for $\wDistVec$, because it returns fresh samples form the next block, given the previously fixed prefix $\pfix{w}{i-1}$.

\remove{
\begin{definition}[Valid prefixes] For a joint distribution $\uDistVec \equiv (\uDist_1,\dots,\uDist_n)$, we define the set of all ``valid'' (i.e., possible) prefixes of $\uDistVec$ as follows: $\VP(\uDistVec) = \set{\pfix{u}{i} \colon i\in [n], \pfix{u}{i} \in \Supp((\uDist_1,\dots,\uDist_i))}$.
\end{definition}
}

\begin{definition} [Online tampering]\label{def:tamp} Let $\wDistVec \equiv (\wDist_1, \dots , \wDist_n)$ be a sequence of jointly distributed random variables, and let $\wDist[\pfix{w}{i-1}]$ be the online sampler for $\wDistVec$ for all $i\in[n]$ and all $\pfix{w}{i-1} \in \Supp(\pfix{\wDist}{i-1})$. Online tampering algorithms for $\wDistVec$ and their properties are defined as follows.
\begin{itemize}
    \item {\bf Online tampering.}  We call a (potentially randomized and  computationally unbounded) algorithm $\Tam$ an  \emph{online tampering} algorithm for $\wDistVec$, if for all $i\in[n]$ and  $\pfix{w}{i} \in \Supp(\pfix{\wDist}{i})$, it holds that
$$\Pr[ \Tam(\pfix{w}{i}) \in \Supp(\wDist[\pfix{w}{i-1}])]=1 ~.$$
Namely, $ \Tam(\pfix{w}{i})$ always outputs a candidate $i\th$ block that still falls into $\Supp(\wDist[\pfix{w}{i-1}])$.

\remove{
    \item {\bf Offline attacks.} For an arbitrary joint distribution $\uDistVec \equiv (\uDist_1\dots,\uDist_n)$, we call a (potentially randomized and possibly computationally unbounded) algorithm $\OffTam$ an  \emph{offline tampering} algorithm for $\uDistVec$, if given any  $\uVec \in \Supp(\uDistVec)$, it holds that
$$\Pr_{\vVec \gets \Tam(\uVec)}[\vVec \in \Supp(\uDistVec)]=1 ~.$$
Namely, given any $\uVec \gets \uDistVec$, $ \Tam(\uVec)$ always outputs a vector in $\Supp(\uDistVec)$. 
}
    \item {\bf Resulting tampered distribution.} For an online tampering algorithm $\Tam$ for $\wDistVec$, by $(\uDistVec,\vDistVec) \equiv \jointSamp{\wDistVec}{\Tam}$ we refer to the jointly distributed sequence of random varaibles defined as follows. 
            For $i=1,2,\dots,n$, we first sample $u_i \gets \wDist[\pfix{v}{i-1}]$, and then we obtain $v_i \gets \Tam(\pfix{v}{i-1},u_i)$ as  the (possibly different than $u_i$) choice of the tampering algorithm $\Tam$ for the $i\th$ block (that will override $u_i$). At the end, we output the pair of sequences $(\uVec=\pfix{u}{n},\vVec = \pfix{v}{n})$ as the sample from $(\uDistVec,\vDistVec)$. 

            {\bf Notation.} For simplicity, we  use $\vDist[\pfix{v}{i-1}]$ to denote $(\vDist_i \mid \pfix{v}{i-1})$ and  use  $(\wDist,\vDist)[\pfix{v}{i-1}]$  to denote  the \emph{jointly} distributed random variables from which  $(u_i,v_i)$ are sampled  conditioned on the prefix $\pfix{v}{i-1}$. The notation allows us to use $\vDist[\pfix{v}{i-1}],(\wDist,\vDist)[\pfix{v}{i-1}]$ similarly to how we use online samplers.\footnote{Note that are \emph{not} defining a similar notation of the form $\uDist[\pfix{v}{i-1}]$ for $\uDistVec$. Firstly, this is not needed as  $\wDist[\pfix{v}{i-1}]$ already  provides a sampler for $u_i$. Moreover, such notation would be inconsistent with our notation for online samplers for random processes based on joint distributions, because the notation would  implicitly interpret  $\pfix{v}{i-1}$ as previous samples from $\pfix{\uDist}{i-1}$.}

    \item {\bf Budget of tampering attacks.} Let $\metric$ be a metric defined  over  $\Supp(\uDistVec)$ as vectors of dimension $n$. We say a tampering algorithm $\Tam$ has \emph{budget} (at most) $b$, if
    $$\Pr_{(\uVec,\vVec) \gets \jointSamp{\wDistVec}{\Tam}}[\metric(\uVec,\vVec) \leq b] = 1.$$
    We say that $\Tam$ has     \emph{average budget} (at most) $b$, if  the following weaker condition holds
    $$\Ex_{(\uVec,\vVec) \gets \jointSamp{\wDistVec}{\Tam}}[\metric(\uVec,\vVec)] \leq b.$$
     \item {\bf Algorithmic efficiency of attacks.}
    If $\wDistVec = \wDistVec_n$ is a member from a \emph{family} defined for all $n \in \N$, 
we call an online or offline tampering algorithm \emph{efficient}, if its running time is $\poly(N)$ where $N$ is the total bit-length representation of  any $\wVec \in \Supp(\wDistVec_n)$.
\end{itemize}
\end{definition}

\begin{definition}[Partial expectations]
\label{defs:partial} Suppose $f \colon \Supp(\wDistVec) \To \R$ for $\wDistVec \equiv (\wDist_1,\dots,\wDist_n)$, $i \in [n]$, and  $\pfix{w}{i}\in \Supp(\pfix{\wDist}{i})$. Then (using a small hat) we define the notation $\fhat{w}{i}= \Ex_{\wVec\gets ({\wDistVec} \mid \pfix{w}{i})}[f(\wVec) ]$ to define the expected value of $f$ for a sample from $\wDistVec$ given the prefix $\pfix{w}{i}$. In particular, for $\wVec=\pfix{w}{n}$, we have $\fhat{\wVec}{}=f(\wVec)$, and also $\fhat{\es}{} = \Ex[f(\wDistVec)].$
\end{definition}

\remove{
\begin{lemma}[Theorem 1 in \cite{becker2012variance}]\Snote{I think we don't need this lemma anymore.} \label{lem:InvJen} Let $\varphi(\cdot)$ be twice differentiable, and let $M = \sup_{x \in [a,b]} \varphi''(x)$. Let $\bfh$ be a random variable where $\Supp(\bfh) \se [a,b]$, and let $\sigma^2 = \Var[\bfh]$. Then it holds that 
$$\Ex[\vphi(\bfh)] - \vphi(\Ex[\bfh]) \leq M \cdot \sigma^2 / 2.$$
\end{lemma}
}

\begin{lemma}[Hoeffding's lemma]\label{lemma:hoefding_lemma} 
Let $\bfx$ be a random variable such that $\Pr[a\leq\bfx\leq b]= 1$ and $\Ex\left[\bfx\right] = 0$. Then, it holds that
$\Ex[e^\bfx] \leq e^{(b-a)^2/8}$.
\end{lemma}
\remove{
\begin{lemma}\label{lemma:numerical}
Let $\bfx$ be a random variable where $\Pr[e^{-\lambda}\leq\bfx\leq e^\lambda]=1$. Then,  $\Ex\left[\ln(\bfx)\right] \geq \ln(\Ex\left[\bfx\right]) -\lambda^2/2$.
\end{lemma}
\begin{proof}
Let $\Ex\left[\ln(\bfx)\right] = s$. Consider a random variable $\bfy = \ln(\bfx) - s$. We have $\Ex[\bfy] = 0$ and $-\lambda -s\leq \bfy\leq \lambda -s$. Therefore, by Lemma \ref{lemma:hoefding_lemma} we have
$$\Ex\left[e^{\bfy}\right] \leq e^{\lambda^2/2}.$$
On the other hand, we have $\Ex\left[e^{\bfy}\right] = \Ex\left[e^{\ln(\bfx) - s}\right] = \Ex\left[\bfx\right]\cdot e^{-s}.$ Therefore, we have
$\Ex\left[\bfx\right]\cdot e^{-s} \leq e^{\lambda^2/2}$ which implies $e^{-s} \leq e^{\lambda^2/2 - \ln(\Ex\left[\bfx\right])}$, and so $s\geq \ln(\Ex\left[\bfx\right]) -\lambda^2/2$. 
\end{proof}
}
\begin{lemma}\label{lemma:numerical2}
Let $\bfx$ be a random variable where $\Pr\left[e^{-\lambda}\leq\bfx\right]=1$ and $\Pr\left[\bfx\leq e^\lambda\right]\geq 1-\delta$ and $\Pr\left[\bfx\leq c\right]=1$. Then,  $\Ex\left[\ln(\bfx)\right] \geq \ln(\Ex\left[\bfx\right]-\delta\cdot c) -\lambda^2/2 $.
\end{lemma}
\begin{proof}
Let $\Ex\left[\min(\ln(\bfx), \lambda)\right] = s$. Consider a random variable $\bfy \equiv \min(\ln(\bfx),\lambda) - s$. We have $\Ex[\bfy] = 0$ and $-\lambda -s\leq \bfy\leq \lambda -s$. Therefore, by Lemma \ref{lemma:hoefding_lemma} we have
$$\Ex\left[e^{\bfy}\right] \leq e^{\lambda^2/2}.$$
On the other hand, we have $\Ex\left[e^{\bfy}\right] = \Ex\left[e^{\min(\ln(\bfx),\lambda) - s}\right] = \Ex\left[\min\left(\bfx,e^\lambda\right)\right]\cdot e^{-s}.$ Thus, we have
$\Ex\left[\min\left(\bfx,e^\lambda\right)\right]\cdot e^{-s} \leq e^{\lambda^2/2}$ which implies $e^{-s} \leq e^{\lambda^2/2 - \ln\left(\Ex\left[\min\left(\bfx,e^\lambda\right)\right]\right)}$, and so $s\geq \ln\left(\Ex\left[\min\left(\bfx,e^\lambda\right)\right]\right) -\lambda^2/2$. Therefore we have, $s\geq \ln\left(\Ex\left[\bfx\right] -\delta \cdot c\right) - \lambda^2/2$.
\end{proof}
\begin{lemma}\label{lemma:numerical3}
Let $\bfx$ be a random variable where $\Pr[e^{-\lambda}\leq\bfx]=1$ and $\Pr[\bfx\leq e^\lambda]\geq 1-\delta$ and $\Pr[\bfx\leq c]=1$. Then,  $\Ex[1/\bfx] \leq  \frac{e^{\lambda^2}}{\Ex[\bfx] - \delta\cdot c}$.
\end{lemma}
\begin{proof}
Let $\Ex[\min(\ln(\bfx), \lambda)] = s$. Consider a random variable $\bfy = \min(\ln(\bfx),\lambda) - s$. Similar to proof of Lemma $\ref{lemma:numerical2}$ we have $s\geq \ln(\Ex[\bfx] -\delta \cdot c) - \lambda^2/2$. Now consider another random variable $\bfy' \equiv -\bfy$. Again by using Hoeffding Lemma we have $\Ex[e^{\bfy'}] \leq e^{\lambda^2/2}$ which means 
$$\Ex[e^{-\min(\ln(\bfx), \lambda)}]\cdot e^s \leq e^{\lambda^2/2}$$
which implies 
$$\Ex[\max(1/\bfx, e^{-\lambda})] \leq e^{\lambda^2/2} \cdot e^{-s} \leq \frac{e^{\lambda^2}}{\Ex[\bfx] -\delta\cdot c}.$$
\end{proof}

The following lemma is implied by Theorem 3.13 from \cite{mcdiarmid1998probabilistic}.
\begin{lemma}[Azuma's inequality for sub-martingales] \label{lem:AzumaApp} Let $\tDistVec \equiv (\tDist_1,\dots,\tDist_n)$ be a sequence of $n$ jointly distributed  random variables such that for all $i \in [n]$, $\Pr[|\tDist_i| \leq c_i]\geq 1-\xi$, for all $\pfix{t}{i-1}\gets \pfix{\tDist}{i-1}$, and that
$\Ex[\tDist_i\mid \pfix{t}{i-1}] \geq -\gamma_i.$
If $\gamma = \sum_{i=1}^n \gamma_i$, then we have
$$\Pr\left[\sum_{i=1}^n \tDist_i \leq -s \right] \leq \e^{\frac{-(s - \gamma)^2}{2 \sum_{i=1}^{n}c_i^2}} + n\cdot \xi.$$
\end{lemma}
\section{Optimal Computational Concentration for Hamming Distance}
In this section, we formally state and prove our main result, which is the computational concentration of measure in any product space under Hamming distance.

\begin{definition}[Weighted Hamming Distance] 
For $\alVec = (\alpha_1,\dots,\alpha_n)\in \R_+^n$, the $\alVec$-weighted Hamming distance between vectors of dimension $n$ is denoted by $\alHD(\cdot,\cdot)$ and is defined as 
$$\alHD(\uVec,\vVec) = \sum_{i \in [n], u_i \neq v_i} \alpha_i.$$
\end{definition}

\begin{theorem} \label{thm:main}
Let $(\alpha_1,\dots,\alpha_n)\in \R^n$ be such that $\sum_{i=1}^n \alpha_i^2 = n.$ Then,  there is a (uniform) oracle-aided randomized algorithm $\Tam$ such that the following holds.
Suppose $f \colon \Supp(\wDistVec) \To \bits$ is a Boolean function for random variable $\wDistVec \equiv (\wDist_1,\dots,\wDist_n)$, and that $ \Pr[f(\wDistVec)=1] = \eps  $. Then,  the oracle-aided algorithm $ \Tam^{\wDist[\cdot],f(\cdot)}(\eps,\delta,\cdot)$ (also denoted by $\Tam$ for simplicity) with access to the online sampler  $\wDist[\cdot]$ for $\wDistVec$ and $f (\cdot)$ as oracles is an online tampering algorithm for $\wDistVec$ and has the following features:
\begin{enumerate}
    \item  $\Pr[f(\vDistVec)=1] \geq 1-\delta$ where $\vDistVec$ is the tampered sequence, i.e., $\Pr_{(\uVec,\vVec) \gets \jointSamp{\wDistVec}{\Tam}}[f(\vVec)=1] \geq 1-\delta $.

    \item $\Tam$'s tampering budget in $\alpha$-weighed Hamming distance $\alHD$ is  $O( \sqrt{n \cdot \ln(\nicefrac{1}{\eps \delta}) })$.
    
    \item $\Tam$ runs in  time $\poly(\nicefrac{N}{\eps  \delta})$ where $N$ is the total bit representation of any $\wVec \gets \wDistVec$.
    
\end{enumerate}
\end{theorem}

\begin{remark}[Corollary for product distributions] If the original random variable $\wDistVec = (\wDist_1,\dots,\wDist_n)$ in Theorem \ref{thm:main} is a product, $\wDistVec = (\wDist_1 \times \dots \times \wDist_n)$, then the distribution of the samples $\uDist$ obtained through $(\uVec,\vVec) \gets \jointSamp{\wDistVec}{\Tam}$  would be identical to that of $\wDistVec$. Namely, we can simply think of the samples $\uVec$ as the original \emph{untampered} vector sampled from $\wDistVec$, and  $\vVec$ would be the  perturbed vector.
\end{remark}

In the rest of this section, we prove  Theorem \ref{thm:main}.

\subsection{Proof Using Promised Approximate Partial Expectation Oracles}

The following result works in the model where the approximate partial-expectations oracle $\ftild{\cdot}{}$ is available to the online tampering algorithm $\AppTam$. 

Consider three oracles $\ftild{v}{i}$ , $\mfix{v}{i}$ and $\ftildstar{v}{i} = \ftild{\pfix{v}{i}, \mfix{v}{i}}{}$ with the guarantee that for all $\pfix{v}{i}\in \Supp(\wDist_{\leq i})$ we have 5 conditions:
\begin{enumerate}
    \item  $\left|\ln{\ftild{v}{i}} - \ln{\fhat{v}{i}}\right| \leq \gamma$,
    \item  $\ftildstar{v}{i} = \ftild{\pfix{v}{i}, \mfix{v}{i}}{} \geq \ftild{v}{i}$,
    \item  $\Pr\left[\ftild{\pfix{v}{i}, \wDist[\pfix{v}{i}]}{} \geq \ftildstar{v}{i}\right] \leq \gamma\cdot \ftild{v}{i}$,
    \item  $0 \leq \ftild{v}{i}\leq 1$,
    \item $\ftild{v}{n} = f(\pfix{v}{n}).$
\end{enumerate}
The first condition states that the approximate partial expectation oracle has a small multiplicative error. The second and third conditions state that $\mfix{v}{i-1}$ is a good approximation of some $v*$ that maximized $\ftild{\pfix{v}{i-1},v^*}{}$.  Now we construct an algorithm using these oracles. 
\begin{construction}[Online tampering using promised \emph{approximate} partial-expectations oracle] \label{const:tampBlock_app}
Recall that we are given a prefix $\pfix{v}{i-1}$ that is finalized, and we are also given a candidate value $u_i$ for the $i$'th block (supposedly sampled from  $\wDist[\pfix{v}{i-1}]$) and we want to decide to keep $v_i=u_i$ or change it.
 Let $\lambda>0$ be a parameter of the attack to be chosen later, $v^*_i = \ftildstar{v}{i-1}$ and let  $\tilde{f^*}=\ftildstar{v}{i-1}$ be that maximum.
\begin{enumerate}
    \item (Case 1)  If $\tilde{f^*} \geq e^{\lambda \alpha_i} \cdot \ftild{\pfix{v}{i-1}}{}$, then output $v_i = v_i^*$ (regardless of $u_i$).
    \item (Case 2) Otherwise, if $\ftild{\pfix{v}{i-1},u_i}{}\leq e^{-\lambda \alpha_i} \cdot \ftild{\pfix{v}{i-1}}{} $ , then output $v_i = v_i^*$.
    \item (Case 3) Otherwise keep the value $u_i$ and output $v_i = u_i$.
\end{enumerate}
\end{construction}
\begin{claim}[Average case analysis of  Construction \ref{const:tampBlock_app}] \label{clm:tampBlock_app} Let $\bfk_i$ be the Boolean random variable that $\bfk_i=1$ iff the tampering over the $i$'th block happens, and let  $\alk = \sum_{i \in [n]} \alpha_i \cdot \bfk_i$ capture  the resulting $\alHD$ distance between the jointly sampled $\uVec$ and $\vVec$. Also let $\tilde{\eps}= \ftild{\es}{}$. Then, it holds that
    $$\ln (1/\tilde{\eps}) \geq \Ex[\alk] \cdot \lambda - {\lambda^2  n}/{2} + n\cdot \ln(1-3\gamma).$$
\end{claim}
\paragraph{Corollary of Claim \ref{clm:tampBlock_app}.} By choosing  $\lambda = \sqrt {2\ln (1/\eps) / n}$, we obtain $\Ex[\alk] \leq \sqrt{2n\ln(1/\eps)}$.

\paragraph{}We prove the following stronger statement that implies Claim \ref{clm:tampBlock_app}.

\begin{claim} \label{clm:stronger_app}
Let $\pfix{v}{i-1}$ be fixed.  Then,
$$\ln (1/\ftild{v}{i-1}) - \Ex_{v_i \gets \vDist[ \pfix{v}{i-1}]}\left[\ln\left(1/\ftild{v}{i}\right)\right] \geq \Pr[\bfk_i] \cdot ( \alpha_i\lambda)  -  \frac{\alpha_i^2 \lambda^2}{2} + \ln(1-3\gamma).$$
\end{claim}
\begin{proof}[Proof of Claim \ref{clm:tampBlock_app} using Claim \ref{clm:stronger_app}.] A key property of Construction \ref{const:tampBlock_app} is that, because the tampering algorithm does not allow the function reach $0$, the final sequence $\vVec$ always makes the function $1$, namely
\begin{equation} \label{eq:always1_app}
    \Pr[f(\pfix{\vDist}{n})=1]=1.
\end{equation} Using the above equation, Claim \ref{clm:tampBlock_app} follows from Claim \ref{clm:stronger_app} and linearity of expectation as follows.
\begin{align*}
    \ln (1/\tilde{\eps}) &= \ln (1/\ftild{\es}{}) - \Ex[\ln (1)]  \\
    \text{(by Equation \ref{eq:always1_app})}~~~~ &= \ln (1/\ftild{\es}{}) - \Ex[\ln (1/\ftild{\vDist}{n})] \\
    \text{(by linearity of expectation)}~~~~ &= \sum_{i \in [n]} \left[\Ex[\ln (1/\ftild{\vDist}{i-1}) - \Ex[\ln (1/\ftild{\vDist}{i})]\right] \\
    \text{(by Claim \ref{clm:stronger_app})}~~~~ &\geq \sum_{i \in [n]} \left[ \Ex[\alpha_i\cdot \bfk_i] \cdot \lambda  - \frac{ \alpha_i^2\cdot \lambda^2}{2} + \ln(1-3\gamma)\right] \\
    \text{(by linearity of expectation)}~~~~ &= \Ex[\alk] \cdot \lambda - \frac{n\cdot\lambda^2}{2} + \ln(1-3\gamma)\cdot n.
\end{align*}
\end{proof}
Now we prove Claim \ref{clm:stronger_app}.
\begin{proof}[Proof of  Claim \ref{clm:stronger_app}]
 There are two cases:
\begin{itemize}
    \item If tampering of Case 1 happens, then we have $\Pr[\bfk_i=1]=1$, and 
    \begin{align*}
       \ln (1/\ftild{v}{i-1}) - \Ex_{v_i \gets (\vDist[ \pfix{v}{i-1}]}[\ln(1/\ftild{v}{i})] \geq \ln(1/\ftild{v}{i-1}) - \ln(1/\tilde{f^*})
       \geq \lambda \alpha_i
    \end{align*}
    Thus, in this case Claim \ref{clm:stronger_app} follows trivially.
    \item If tampering of Case 1 does not happen, it means that $\tilde{f^*}$  is bounded from above. In the following, we focus on this case and all the probabilities and expectations are conditioned on Case 1 not happening; namely, we have $\tilde{f^*} \leq \ftild{v}{i-1}\cdot e^{\lambda}$ .
\end{itemize}

Let $I(\pfix{v}{i})$ be the indicator function for the set $\set{\pfix{v}{i}\colon {\ftild{v}{i}} \leq e^{-\lambda \alpha_i}\cdot{\ftild{v}{i-1}}}.$ We have
$$\Pr_{(u_i,v_i) \gets (\wDist,\vDist)[\pfix{v}{i-1}]}\left[\ftild{v}{i} \geq \max\left(e^{-\lambda \alpha_i}\cdot \ftild{v}{i-1},\ftild{\pfix{v}{i-1},u_i}{}\right)\cdot e^{\lambda \alpha_i\cdot I(\pfix{v}{i-1},u_i)} \right] = 1.$$

This is correct because we are either in Case 2, which means $I(\pfix{v}{i}) = 1$ and
$$\ftild{v}{i} = \tilde{f^*} \geq \ftild{v}{i-1} \geq \ftild{\pfix{v}{i-1}, u_i}{}\cdot e^{\lambda\cdot\alpha_i}$$
or we are in Case 3 which means $I(\pfix{v}{i})=0$ and 
$$\ftild{v}{i} = \ftild{\pfix{v}{i-1}, u_i}{}.$$
Note that the two terms on each side of the inequality inside the probability above depend only on either of $u_i$ or $v_i$ (not both). Therefore, by linearity of expectation  we have

 \begin{align*}
 &\Ex_{v_i \gets \vDist[\pfix{v}{i-1}]}[\ln(\ftild{v}{i})] \\
 &\geq \Ex_{u_i \gets \wDist[\pfix{v}{i-1}]}\left[\ln\left(\max\left(e^{-\lambda\alpha_i}\cdot \ftild{v}{i-1},\ftild{\pfix{v}{i-1},u_i}{}\right)\cdot e^{\lambda\cdot\alpha_i\cdot I(\pfix{v}{i-1},u_i)}\right) \right]\\
 &= \Ex_{u_i \gets \wDist[\pfix{v}{i-1}]}\left[\ln\left(\max\left(e^{-\lambda\alpha_i}\cdot \ftild{v}{i-1},\ftild{\pfix{v}{i-1},u_i}{}\right)\right)\right] + \lambda\cdot\alpha_i \cdot \Ex_{u_i \gets \wDist[\pfix{v}{i-1}]}[ I(\pfix{v}{i-1},u_i) ] \\
&= \Ex_{u_i \gets \wDist[\pfix{v}{i-1}]}\left[\ln\left(\max\left(e^{-\lambda\alpha_i}\cdot \ftild{v}{i-1},\ftild{\pfix{v}{i-1},u_i}{}\right)\right)\right] + \lambda\cdot\alpha_i \cdot \Ex[\bfk_i]
. \stepcounter{equation}\tag{\theequation}\label{ineq:03}
 \end{align*}
Now consider the random variable $\bft$ for a fixed $\pfix{v}{i-1}$ as follows $$\bft\equiv\frac{ \max\left(e^{-\lambda}\cdot \ftild{v}{i-1},\ftild{\pfix{v}{i-1},\wDist[\pfix{v}{i-1}]}{}\right)}{\ftild{v}{i-1}}.$$
It holds that
\begin{equation}
    \Pr\left[e^{-\lambda\alpha_i}\leq \bft\right]=1.\label{eq:001}
\end{equation}
We also know by condition 3 of the $\tilde{f^*}(\cdot)$ oracle that 
$$\Pr[\ftildstar{v}{i-1} \geq \ftild{\pfix{v}{i-1},\wDist[\pfix{v}{i-1}]}{}] \geq 1-\gamma\cdot \ftild{v}{i-1}$$ 
which together with $\ftildstar{v}{i-1}\leq \ftild{v}{i-1}\cdot e^{\lambda\alpha_i} $ implies 
\begin{equation}\label{eq:002}
\Pr[\bft\leq e^{\lambda\cdot\alpha_i}] \leq 1-\gamma\cdot \ftild{v}{i-1}.
\end{equation}
We also know that 
\begin{equation}\label{eq:003}
    \Pr\left[\bft\leq \frac{1}{\ftild{v}{i-1}}\right] =1.
\end{equation}
We also have
\begin{align*}
\Ex\left[\bft\cdot \ftild{v}{i-1}\right] &= \Ex \left[\max\left(e^{-\lambda}\cdot \ftild{v}{i-1},\ftild{\pfix{v}{i-1},\wDist[\pfix{v}{i-1}]}{}\right)\right]\\
&\geq \Ex \left[\ftild{\pfix{v}{i-1},\wDist[\pfix{v}{i-1}]}{}\right] \\
&\geq \Ex \left[\fhat{\pfix{v}{i-1},\wDist[\pfix{v}{i-1}]}{}\right]\cdot e^{-\gamma}\\
&= \fhat{v}{i-1}\cdot e^{-\gamma}\\
&\geq \ftild{v}{i-1}\cdot e^{-2\gamma}.
\end{align*}
which implies
\begin{equation}\label{eq:004}
    \Ex[\bft]\geq e^{-2\gamma}\geq 1-2\gamma.
\end{equation} 
Therefore using \ref{eq:001}, \ref{eq:002}, \ref{eq:003} and \ref{eq:004} and applying Lemma \ref{lemma:numerical2} we get,
\begin{align*}
\Ex[\ln(\bft)] \geq \ln\left(\Ex[\bft]- \gamma\cdot\ftild{v}{i-1}\cdot \frac{1}{\ftild{v}{i-1}}\right) - \frac{\alpha_i^2\cdot\lambda^2}{2}  \geq \ln(1 -3\gamma) -\frac{\alpha_i^2\cdot\lambda^2}{2}.\stepcounter{equation}\tag{\theequation}\label{ineq:04}
\end{align*}
 Combining Equations (\ref{ineq:03}) and (\ref{ineq:04}), we get
 $$\Ex_{v_i \gets \vDist[\pfix{v}{i-1}]}\left[\ln(\ftild{v}{i})\right]  \geq \ln(\ftild{v}{i-1}) + \lambda\cdot\alpha_i \cdot \Ex[\bfk_i] - \frac{\lambda^2\cdot \alpha_i^2}{2} +\ln(1-3\gamma)$$
 which finishes the proof.
\end{proof}

\begin{claim}[Worst case analysis of  Construction \ref{const:tampBlock_app}] \label{clm:tampBlock_app_worst} Let $\bfk_i$ be the Boolean random variable that $\bfk_i=1$ iff the tampering over the $i$'th block happens, and let  $\alk = \sum_{i \in [n]} \alpha_i \cdot \bfk_i$ capture  the resulting $\alHD$ distance between the jointly sampled $\uVec$ and $\vVec$. Also let $\tilde{\eps}= \ftild{\es}{}$. Then, it holds that
    $$\Pr[\bfK \geq k] \leq \frac{e^{  (\sum_{i=1}^{n} \alpha_i^2)  \lambda^2-    k \lambda}}{\tilde{\eps}\cdot (1-2\gamma)^n}.$$
\end{claim}
\begin{proof}
We prove this claim by induction on $n$. Let $A(n,k,\tilde{\eps})$ be a function that indicates the maximum probability of using more than $k$ budget, over all random processes with boolean outcome of length $n$,  and average $\tilde{\eps}$. We want to inductively show that 
$$A(n,k,\tilde{\eps})\leq \frac{e^{  (\sum_{i=1}^{n} \alpha_i^2)\cdot\lambda^2-    k \lambda}}{\tilde{\eps}\cdot (1-2\gamma)^n}.$$
Consider different cases that might happen during the tampering of first block. If we tamper on first block through Case I, we have
$$\Pr[\bfK \geq k] \leq A(n-1, k-\alpha_1, \ftildstar{\es}{})$$
And by induction hypothesis we have
$$A(n-1, k-\alpha_1, \ftildstar{\es}{})\leq \frac{e^{  (\sum_{i=2}^{n} \alpha_i^2)  \lambda^2 -k \lambda+\lambda\cdot\alpha_i}}{\ftildstar{\es}{}\cdot (1-2\gamma)^n} \leq \frac{e^{  (\sum_{i=2}^{n} \alpha_i^2)  \lambda^2-    k \lambda + \lambda\cdot\alpha_i}}{e^{\lambda\alpha_1}\cdot\tilde{\eps} \cdot (1-2\gamma)^n}\leq \frac{e^{(\sum_{i=1}^{n} \alpha_i^2)\cdot \lambda^2-    k \lambda}}{\tilde{\eps}\cdot (1-2\gamma)^n}.$$
So the induction goes through for Case 1. If we are not in Case 1, then we have,
\begin{align*}
\Pr\left[\bfK \geq k\right] &= \Pr\left[\bfK \geq k \mid\text{ Case 3}\right]\cdot \Pr\left[\text{Case 3}\right] + \Pr\left[\bfK \geq k \mid\text{ Case 2}\right]\cdot \Pr\left[\text{Case 2}\right]\\
&\leq \Ex\left[A\left(n-1, k, \ftild{u}{1}\right) \mid \text{ Case 3}\right]\cdot \Pr\left[\text{Case 3}\right]\\
&~~~~~+ \Ex\left[A\left(n-1, k-\alpha_1, \ftildstar{\es}{}\right) \mid \text{ Case 2}\right]\cdot \Pr\left[\text{Case 2}\right]\\
&\leq \Ex\left[  \frac{e^{ (\sum_{i=2}^{n} \alpha_i^2) \cdot\lambda^2-    k \lambda}}{\ftild{u}{1} - 2(n-1)\gamma}\mid \text{ Case 3}\right]\cdot \Pr\left[\text{Case 3}\right]\\
&~~~~~+ \Ex\left[\frac{e^{ (\sum_{i=2}^{n} \alpha_i^2) \cdot\lambda^2-    k \lambda + \lambda\cdot \alpha_1}}{\ftildstar{\es}{}\cdot (1-2\gamma)^{n-1}} \mid \text{ Case 2}\right]\cdot \Pr\left[\text{Case 2}\right]\\
&\leq \Ex\left[  \frac{e^{ (\sum_{i=2}^{n} \alpha_i^2) \cdot\lambda^2-    k \lambda}}{\ftild{u}{1}\cdot (1-2\gamma)^{n-1}}\mid \text{ Case 3}\right]\cdot \Pr\left[\text{Case 3}\right]\\
&~~~~~+ \Ex\left[\frac{e^{ (\sum_{i=2}^{n} \alpha_i^2) \cdot\lambda^2-    k \lambda + \lambda\cdot \alpha_1}}{\max\left(e^{-\lambda\cdot\alpha_i}\cdot\tilde{\eps}, \ftild{u}{1}\right)\cdot e^{\lambda\cdot \alpha_i}\cdot (1-2\gamma)^{n-1}} \mid \text{ Case 2}\right]\cdot \Pr\left[\text{Case 2}\right]\\
&\leq e^{(\sum_{i=2}^{n} \alpha_i^2)\cdot \lambda^2 - k\lambda} \cdot \Ex\left[\frac{1}{\max(e^{-\lambda\cdot\alpha_i}\cdot\tilde{\eps}, \ftild{u}{1})\cdot (1-2\gamma)^{n-1}}\right]\stepcounter{equation}\tag{\theequation}\label{eq:0001}
\end{align*}
We know that $\Ex[\max(e^{-\lambda\cdot\alpha_i}\cdot\tilde{\eps}, \ftild{u}{1}))] \geq \Ex[\ftild{u}{1}] \geq \tilde{\eps}\cdot e^{-\gamma}$. Now we can use Lemma \ref{lemma:numerical3} and get
\begin{equation}\label{eq:0002}\Ex[\frac{1}{\max(e^{-\lambda\cdot\alpha_1}\cdot\tilde{\eps}, \ftild{u}{1}))}] \leq \frac{e^{\alpha_1^2\cdot\lambda^2}}{\tilde{\eps}\cdot (e^{-\gamma} - \gamma)} \leq \frac{e^{\alpha_1^2\cdot\lambda^2}}{\tilde{\eps} \cdot (1-2\gamma)}\end{equation}
Combining Equations \ref{eq:0001} and \ref{eq:0002} we get,
$$\Pr[\bfK \geq k] \leq \frac{e^{(\sum_{i=1}^{n} \alpha_i^2)\lambda^2 - k\lambda}}{\tilde{\eps}(1 - 2\gamma)^n}$$
which finishes the proof. 
\end{proof}
\subsubsection{Tampering  with Abort}
The Construction~\ref{const:tampBlock_app} achieves average close to 1 with small number of tampering. However we cannot implement that construction it in polynomial time. The problem is that it is hard to instantiate the oracle $\tilde{f}(\cdot)$ and $\tilde{f^*}(\cdot)$ in polynomial time when the partial average gets close to $0$. Following we add a step to our construction to address this issue. Then we will show that this additional step will not hurt the performance of the algorithm by much. 
\begin{construction}[Online tampering \emph{with abort} $\AppTamAb$ using promised approximate partial-expectations oracle] \label{const:tampBlockAbort_app} This construction is identical to Construction~\ref{const:tampBlock_app}, except that whenever the  fixed prefix has a too small approximate partial expectation $\ftild{\pfix{v}{i-1},u_i}{}$ (based on a parameter $\tau$) we will abort. Also, in that case the tampering algorithm does not tamper with any future $v_i$ block either. Namely, we add the following ``Case 0'' to the previous steps:
\begin{itemize}
\item (Case 0) If $\ftild{\pfix{v}{i-1},u_i}{}\leq e^{-\tau}\cdot \tilde{\eps}$ abort ($\tilde{\eps} = \ftild{\es}{}$). If had aborted before, do nothing. 
\end{itemize}
\end{construction}
\paragraph{Average and worst case analysis of  Construction~\ref{const:tampBlockAbort_app}.} The average number of tampering of Construction~\ref{const:tampBlockAbort_app} is trivially less than average number of tampering of Construction~\ref{const:tampBlock_app}. Therefore, the same bound of Claim~\ref{clm:tampBlock_app} still applies to Construction~\ref{const:tampBlockAbort_app} as well. Also, the probability of number of tampering going beyond some threshold does not increase compared to Construction~\ref{const:tampBlock_app} which means the same bound of Claim~\ref{clm:tampBlock_app_worst} hold here.

\begin{claim}\label{clm:AbortAverage_app} The probability of ever aborting during sampling 
$(\uVec,\vVec)\gets \jointSamp{\wDistVec}{\TamAb}$  
is at most $n\cdot e^{-\frac{(\tau - n\cdot\lambda^2/2)^2}{2\cdot n\cdot \lambda^2}}$. As a result,  we also have
$$\Ex_{(\uVec,\vVec)\gets \jointSamp{\wDistVec}{\TamAb}}[f(\vVec)] \geq 1 - n\cdot e^{-\frac{(\tau - n\cdot\lambda^2/2)^2}{2\cdot n\cdot \lambda^2}} - n^2\gamma.$$ 
\end{claim}
\begin{proof}
Define Boolean indicator functions $I_0, I_1, I_2$ and $I_3$, as well as a real-valued vector $\yVec$  as follows.
The first function $I_0$ indicates that we have not aborted yet, and the others define a condition for their corresponding cases in Construction~\ref{const:tampBlock_app}.
$$I_0(\pfix{v}{i-1}) = \begin{cases}
0 & \text{if $\forall j\leq i;~\ftild{v}{j} \geq \tilde{\eps}\cdot e^{-\tau}$,}\\
1 & \text{otherwise.}\\
\end{cases}$$
$$I_1(\pfix{v}{i-1})= \begin{cases}
1 & \text{if $\ftildstar{v}{i-1}\geq e^{\lambda\cdot \alpha_i}\cdot\ftild{v}{i-1}$ and $\neg I_0(\pfix{v}{i-1})$},\\
0 & \text{otherwise;}\\
\end{cases}$$
$$I_2(\pfix{v}{i}) =\begin{cases}
1 & \text{if $\ftild{v}{i} \leq e^{-\lambda\cdot \alpha_i}\cdot\ftild{v}{i-1}$ and $\neg I_1(\pfix{v}{i-1})$ and $\neg I_0(\pfix{v}{i-1})$},\\
0 & \text{otherwise;}\\
\end{cases}
$$
The last function indicates that the above conditions are \emph{not} happening.
$$I_3(\pfix{v}{i})=\begin{cases}
1 & \text{if $\neg I_0(\pfix{v}{i})$ and $\neg I_1(\pfix{v}{i-1})$ and $\neg I_2(\pfix{v}{i})$},\\
0 & \text{otherwise.}\\
\end{cases}$$

Finally, we define a real-valued function $y$ as follows that captures the change in the potential function for the cases where none of $I_0,I_1,I_2$ are happening.
$$y(\pfix{v}{i}) = \left(\ln(\ftild{v}{i}) - \ln(\ftild{v}{i-1})\right)\cdot I_3(\pfix{v}{i}).$$
Now consider a sequence of random variables $\yDistVec = (\yDist_1,\dots,\yDist_n)$ sampled as follows. We first sample $(\uVec, \vVec )\gets (\uDistVec, \vDistVec)$ then set $y_i= y(\pfix{v}{i-1}, u_i) = y(\pfix{v}{i})$ for $i \in [n]$. Note that $y(\pfix{v}{i-1}, u_i) = y(\pfix{v}{i})$ because if $I_3(\pfix{v}{i-1}, u_i) = 1$ it means that $u_i = v_i$.
\begin{claim}\label{clm:submart1_app}
We have $\Ex[e^{\yDist_i} \mid \pfix{y}{i-1}] \geq e^{-2\gamma}$.
\end{claim}
\paragraph{Notation.} Since $I_j(\cdot)$'s are Boolean, we can use the notation $(I_i \lor I_j)(\pfix{v}{i})$ or $(1- (I_i \lor I_j))(\pfix{v}{i})$ based on logical operators to construct more Boolean indicators.
\begin{proof}[Proof of Claim~\ref{clm:submart1_app}] 
The high level idea is that $e^{\yDist_i}$ is approximately equal to $\fhat{\pfix{v}{i-1},{u_i}}{}/\fhat{v}{i}$ when we are in Case 3. The average of $\fhat{\pfix{v}{i-1},{u_i}}{}/\fhat{v}{i}$ conditioned on Case 2 and Case 3 is exactly 1. We know that in Case 2 the average is less than one, therefore the average in Case 3 should be at least 1. Following, we formalize this idea.
\begin{align*}
\Ex[e^{\yDist_i}\mid\pfix{y}{i-1}] &= \Ex_{\pfix{v}{i-1} \gets (\pfix{\vDist}{i-1} \mid \pfix{y}{i-1})}\left[\Ex_{u_i \gets \wDist[\pfix{v}{i-1}]}\left[e^{\left(\ln(\ftild{\pfix{v}{i-1},u_i}{}) - \ln(\ftild{v}{i-1})\right)\cdot I_3(\pfix{v}{i-1},u_i)}\right]\right]\\
&\geq\Ex_{\pfix{v}{i-1} \gets (\pfix{\vDist}{i-1} \mid \pfix{y}{i-1})}\left[\Ex_{u_i \gets \wDist[\pfix{v}{i-1}]}\left[e^{\left(\ln(\ftild{\pfix{v}{i-1},u_i}{}) - \ln(\ftild{v}{i-1})\right)\cdot \left((I_3 \lor I_2)\left(\pfix{v}{i-1},u_i\right)\right)}\right]\right]\\
&= \Ex_{\pfix{v}{i-1} \gets (\pfix{\vDist}{i-1} \mid \pfix{y}{i-1})}\left[\Ex_{u_i \gets \wDist[\pfix{v}{i-1}]}\left[e^{\left(\ln(\ftild{\pfix{v}{i-1},u_i}{}) - \ln(\ftild{v}{i-1})\right)\cdot \left(1-(I_1 \lor I_0)\left(\pfix{v}{i-1}\right)\right)}\right]\right]\\
&\geq \Ex_{\pfix{v}{i-1} \gets (\pfix{\vDist}{i-1} \mid \pfix{y}{i-1})}\left[\min\left(\Ex_{u_i \gets \wDist[\pfix{v}{i-1}]}\left[e^{\left(\ln(\ftild{\pfix{v}{i-1},u_i}{}) - \ln(\ftild{v}{i-1})\right)}\right], 1\right)\right]\\
&= \Ex_{\pfix{v}{i-1} \gets (\pfix{\vDist}{i-1} \mid \pfix{y}{i-1})} \left[\min\left(\Ex_{u_i \gets \wDist[\pfix{v}{i-1}]}\left[\ftild{\pfix{v}{i-1},u_i}{}/\ftild{v}{i-1}\right],1\right)\right]\\
&\geq \Ex_{\pfix{v}{i-1} \gets (\pfix{\vDist}{i-1} \mid \pfix{y}{i-1})} \left[\min\left(e^{-2\gamma}\cdot\Ex_{u_i \gets \wDist[\pfix{v}{i-1}]}\left[\fhat{\pfix{v}{i-1},u_i}{}/\fhat{v}{i-1}\right],1\right)\right]\\
&= e^{-2\gamma}.
\end{align*}
\end{proof}

\begin{claim}\label{clm:tbounded_app}
We have $\Pr[\yDist_i \geq -\lambda\cdot \alpha_i] = 1$ and $\Pr[\yDist_i \leq \lambda\cdot\alpha_i] \geq 1-\gamma\cdot\ftild{v}{i-1}$.
\end{claim}
\begin{proof}
If $I_3(\pfix{v}{i})=0$ then $y(\pfix{v}{i})=0$ and both inequalities hold. On the other hand, If $I_3(\pfix{v}{i}) = 1$ it means that $e^{-\lambda\cdot \alpha_i}\cdot \ftild{v}{i-1}\leq \ftild{v}{i}$. Also $\Pr[\yDist_i  \leq e^{\lambda\cdot \alpha_i}\cdot\ftild{v}{i-1}]\geq 1-\gamma\cdot\ftild{v}{i-1}$ holds because of gaurantee of the oracle $\tilde{f^*(\cdot)}$. 
\end{proof}
\begin{claim}\label{clm:submart2_app}
We have $\Ex[\yDist_i \mid \pfix{y}{i-1}] \geq \ln(1-3\gamma) -\frac{\lambda^2\cdot\alpha_i^2}{2}$.
\end{claim}
\begin{proof}
The proof follows by using Lemma~\ref{lemma:numerical2} and Claims~\ref{clm:submart1_app} and~\ref{clm:tbounded_app}.
\end{proof}
\begin{claim}
The probability of aborting is at most 
$n\cdot e^{\frac{(\tau - \lambda n\cdot\lambda^2/2)^2}{2\cdot n\cdot \lambda^2}}.$
\end{claim}
\begin{proof}
By Claims~\ref{clm:submart2_app} and~\ref{clm:tbounded_app}, the sequence $\yDistVec = (\yDist_1,\dots,\yDist_n)$ forms an (approximate) submartingale and by Azuma inequality of Lemma~\ref{lem:AzumaApp} we have,
$$\Pr\left[\sum_{i=1}^n \yDist_i \leq -\tau \right] \leq e^{-\frac{\left(\tau - n\cdot\lambda^2/2\right)^2}{2\cdot n\cdot \lambda^2}} + n\cdot \gamma ~~.$$
On the other hand, for every $\pfix{v}{i} \in \Supp(\pfix{\vDist}{i})$ we have 
$I_2(\pfix{v}{i}) = 0.$ Therefore, for every $\pfix{v}{j} \in \Supp(\pfix{\vDist}{j})$,
\begin{align*}
\ln(\ftild{v}{j}) &= \ln(\tilde{\eps}) + \sum_{i=1}^j (\ln(\ftild{v}{i})-\ln(\ftild{v}{i-1}))\\
 &= \ln(\tilde{\eps}) + \sum_{i=1}^j \left(\ln(\ftild{v}{i})-\ln(\ftild{v}{i-1})\right)\cdot\left((I_0\lor I_1)(\pfix{v}{i-1}) + I_3(\pfix{v}{i})\right)\\
 &\geq \ln(\tilde{\eps}) + \sum_{i=1}^j \left(\ln(\ftild{v}{i})-\ln(\ftild{v}{i-1})\right)\cdot\left(I_0(\pfix{v}{i-1}) + I_3(\pfix{v}{i})\right)\\
 &= \ln(\tilde{\eps}) + \sum_{i=1}^j y(\pfix{v}{i}) + \sum_{i=1}^j \left(\ln(\ftild{v}{i})-\ln(\ftild{v}{i-1})\right)\cdot I_0(\pfix{v}{i-1}).
\end{align*}
We now calculate probability of the event $A_j$ that the partial average goes bellow $e^{-\tau}\cdot \tilde{\eps}$ (i.e., abort happens)  at the $j\th$ block for the \emph{first time}.
\begin{align*}
    &\Pr_{\pfix{v}{j} \gets \pfix{\vDist}{j}}[A_j] \\
    &= \Pr_{\pfix{v}{j} \gets \pfix{\vDist}{j}}\left[ \ftild{v}{j} \leq e^{-\tau}\cdot \tilde{\eps} \land \neg I_0(\pfix{v}{j-1})\right]\\
    &= \Pr_{\pfix{v}{j} \gets \pfix{\vDist}{j}}[ \ln(\ftild{v}{j}) \leq -\tau +\ln(\tilde{\eps}) \land \neg I_0(\pfix{v}{j-1})]\\
    &\leq \Pr_{\pfix{v}{j} \gets \pfix{\vDist}{j}}\left[\sum_{i=1}^j y(\pfix{v}{i}) + \sum_{i=1}^j \left(\ln(\ftild{v}{i})-\ln(\ftild{v}{i-1})\right)\cdot I_0(\pfix{v}{i-1}) \leq -\tau \land \neg I_0(\pfix{v}{j-1})\right]\\
   &\leq \Pr_{\pfix{v}{j} \gets \pfix{\vDist}{j}}\left[\sum_{i=1}^j y(\pfix{v}{i}) \leq -\tau\right] \\
   &\leq e^{-\frac{(\tau - n\cdot\lambda^2/2)^2}{2\cdot n\cdot \lambda^2}} + n\cdot \gamma.
\end{align*}
The above means that the probability that the tampering algorithm of Construction~\ref{const:tampBlockAbort_app} enters the abort state is less than $n\cdot e^{-\frac{(\tau - n\cdot\lambda^2/2)^2}{2\cdot n\cdot \lambda^2}} + n^2\cdot \gamma$. 
\end{proof}
We already know that if abort does not happen then the output will always be $1$. Therefore, we have
$$\Ex_{\vVec\gets \vDistVec}[f(\vVec)] \geq 1- n\cdot e^{-\frac{(\tau - n\cdot\lambda^2/2)^2}{2\cdot n\cdot \lambda^2}} - n^2\gamma.$$
$$\Ex_{\vVec\gets \vDistVec}[f(\vVec)] \geq 1-\delta.$$
\end{proof}

\subsection{Putting Things Together}
In this subsection we show how to instantiate parameters of Construction \ref{const:tampBlockAbort_app} so that we can get polynomial time attack. We first show how to instantiate the oracles. To compute oracle $\ftild{v}{i}$, we sample $\frac{8}{\gamma^3\cdot e^{-\tau}\cdot \tilde{\eps}}$ random continuation and take the average over all of them. By Hoeffding inequality, if $\fhat{v}{i} \geq e^{-\tau}\cdot \tilde{\eps}$ we get the following:
$$\Pr[|\ln(\ftild{v}{i}) -\ln(\fhat{v}{i})| \geq \gamma]\leq \gamma.$$

For $\mfix{v}{i}$ and $\ftildstar{v}{i}$ oracle, sample $\frac{1}{\gamma^2\cdot e^{-\tau}\cdot \tilde{\eps}}$ number of $v_{i+1}$ and take the maximum over $\ftild{v}{i+1}$. This way, we can easily bound the probability of Conditions 2 or 3 not happening by $\gamma$ for all $\pfix{v}{i}$ that $\ftild{v}{i}\geq e^{-\tau}\cdot \tilde{\eps}$. Note that in both of these oracle, we are ignoring the case where $\ftild{v}{i}$ is smaller than the threshold that causes the construction to abort. This enables us to achieve high confidence on our oracles. Using these oracles, we can bound the average of function, average budget and worst case budget of construction \ref{const:tampBlockAbort_app} as follows. Based on Claim \ref{clm:AbortAverage_app} we have
$$\Ex_{(\uVec,\vVec)\gets \jointSamp{\wDistVec}{\TamAb}}[f(\vVec)] \geq 1 - n\cdot e^{-\frac{(\tau - n\cdot\lambda^2/2)^2}{2\cdot n\cdot \lambda^2}} - n^2\gamma - 2n\cdot \gamma.$$
The last $-2n\cdot\gamma$ is added to the right hand side to capture the probability of any of the algorithm's oracle calls failing. For the average budget, following Claim \ref{clm:tampBlock_app} we have,
$$\Ex[\alk]\leq \frac{\ln (1/\tilde{\eps}) + {\lambda^2  n}/{2} - n\cdot \ln(1-3\gamma)}{\lambda} + 2\cdot n\cdot \gamma.$$
And for the worst case budget, following Claim \ref{clm:tampBlock_app_worst} we have 
$$\Pr[\bfK \geq k] \leq \frac{e^{n\lambda^2 - k\lambda}}{\tilde{\eps} - 2\gamma} + 2n\cdot\gamma.$$
\paragraph{Instantiating the Average Case Algorithm:} Now if we set $\lambda = \sqrt{\nicefrac{-2\ln(\eps)}{n}}$, $\tau = \ln(\nicefrac{1}{\tilde{\eps}}) + \sqrt{4\ln(\nicefrac{\delta}{2n})\cdot\ln(\tilde{\eps})}$ and $\gamma = \frac{\delta}{24n^2}$ then we can provide the oracles in time $poly(\nicefrac{n}{\eps\cdot{\delta}})$ and we get:
$$\Ex_{(\uVec,\vVec)\gets \jointSamp{\wDistVec}{\TamAb}}[f(\vVec)] \geq 1 - \delta$$
and 
$$\Ex[\alk]\leq \sqrt{-2n\ln(\eps)} + \delta.$$
\paragraph{Instantiating the Worst Case Algorithm:} Also, for the worst case attacks. If we select the tampering budget $k = \sqrt{2n\cdot\ln(\delta/8)\cdot \ln(\eps/2})$ and then let $\lambda = k/2n$. For $\tau = \ln(\nicefrac{1}{\tilde{\eps}}) + \sqrt{4\ln(\nicefrac{\delta}{2n})\cdot\ln(\tilde{\eps})}$ and $\gamma = \min(\delta/24n^2, \eps/4n)$ we get an algorithm that runs in time $poly(\nicefrac{n}{\eps\cdot{\delta}})$, uses at most $k$ tamperings and increases the average as follows
$$\Ex_{(\uVec,\vVec)\gets \jointSamp{\wDistVec}{\TamAb}}[f(\vVec)] \geq 1 - \delta.$$

\section{Algorithmic Reductions for Computational Concentration} \label{sec:reduc}

In this section, we show a generic framework to prove computational concentration for a metric probability space by reducing its computational concentration to that of another metric probability space. We first define an embedding with some properties. 

\begin{definition}
Let $S_1 = (\cX_1,\metric_1,\msr_1)$ and $S_2= (\cX_2,\metric_2,\msr_2)$ be two metric probability spaces. We call a pair of  mappings $(\fot,\fto)$ (where $\fot$ and $\fto$ are potentially randomized) an $(\alpha,b,w)$ computational concentration (CC) reduction from $S_1$ to $S_2$ if the following hold:
\begin{itemize}
    \item {\bf Probability embedding.} The distribution $\fot(\msr_1)$ is $\alpha$-close (in statistical distance) to $\msr_2$ and $\fto(\msr_2)$ is $\alpha$-close to $\msr_1$.
    \item {\bf Almost Lipschitz property of $\fto$.}  With probability $1$ over  all $x,x' \gets \msr_2$,  $\metric_1(\fto(x),\fto(x')) \leq w\cdot\metric_2(x,x') + b$.
    \item  {\bf Almost inverse mappings.} For every $x_1 \in \cX_1$, and all $x_2 \gets \fot(x_1)$, it holds that  $\metric_1(x_1,\fto(x_2)) \leq b$.
\end{itemize}
\end{definition}


Now we have the following lemma which how to reduce computational concentration on a metric probability space by reducing it to computational concentration on another metric probability space using the embedding between them.
\begin{theorem}\label{thm:reduction}
Let $S_2 = (\cX_2,\metric_2,\msr_2)$ be a metric probability space and let $A_2^{\cS(\cdot)}: \cX_2 \to \cX_2$ be an oracle algorithm such that for any subset $\cS\subseteq\cX_2$ we have $\metric_2(A_2^{\cS(\cdot)}(x),x)\leq k$ and
$$\Pr_{x\gets \msr_2}[A_2^{\cS(\cdot)}(x) \in \cS] \geq c(\msr_2(\cS))$$
for a function $c\colon [0,1] \to [0,1]$. 
If $(\fot,\fto)$ is an $(\alpha,b,w)$ CC reduction from $S_1 = (\cX_1,\metric_1,\msr_1)$ to $S_2= (\cX_2,\metric_2,\msr_2)$, then there is an oracle algorithm $A_1^{\cS(\cdot)}\colon \cX_1 \to \cX_1$ such that for any subset $\cS\subseteq \cX_1$ we have $\metric_1(A_1^{\cS(\cdot)}(x),x)\leq w\cdot k+2b$ and 
$$\Pr_{x\gets \msr_1}[A_1^{\cS(\cdot)}(x) \in \cS] \geq c(\msr_1(\cS)/2-\alpha) -\alpha - \negl(n).$$
 Furthermore, if $A_2$, $\fot$ and $\fto$ run in time $\poly(\frac{n}{\eps})$, then $A_1$ also runs in time $\poly(\frac{n}{\eps})$.
\end{theorem}

\begin{proof}
We define algorithm $A_1^{\cS(\cdot)}$ on input $x$ as follows: $A_1$ first computes $f(x_1)$ to get $x'_1$. Then it creates a set $\cS' = \set{x\in \cX_2 \colon \Pr[g(x) \in \cS] \geq 1/2}$ and runs $A_2^{\cS'(\cdot)}$ on $x'_1$ to get $x'_2$. Then, it computes $g(x'_2)$ for at most $n$ times until it gets some $x_2 \in \cS$ and outputs $x_2$, otherwise it outputs a fresh $g(x'_2)$. We have
\begin{align*}
\Pr_{x_1\gets \msr_1}[A_1^{\cS(\cdot)}(x_1) \in \cS] &\geq 
\Pr_{x_1\gets \msr_1}[A_2^{\cS'(\cdot)}(f(x_1)) \in \cS'] - 2^{-n}\\
&\geq \Pr_{x'_1\gets \msr_2}[A_2^{\cS'(\cdot)}(x'_1) \in \cS']-\alpha - 2^{-n}\\
&\geq c(\msr_2(\cS')) - 2^{-n} - \alpha\\
&\geq c(\msr_1(\cS)/2 - \alpha) - 2^{-n} -\alpha.
\end{align*}
Note that the oracle $\cS'(\cdot)$ cannot be implemented in polynomial time, but it could be approximated with negligible error in polynomial time. On the other hand, we have
\begin{align*}
    \metric_1(A_1(x_1), x_1) &= \metric(x_2, x_1)\\
    &\leq \metric_1(x_2, g(x'_1)) + \metric_1(g(x'_1),x_1)\\
    &\leq \metric_1(x_2, g(x'_1)) + b\\
    &\leq w\cdot \metric_2(x'_2, x'_1)) +2b\\
    &\leq  w\cdot k + 2b.
\end{align*}
\end{proof}


The following construction shows an embedding from Gaussian distribution to hamming cube. Using this embedding and Lemma \ref{thm:reduction} we get computational concentration for the Gaussian distribution. The following embedding uses ideas similar to \cite{bobkov1997isoperimetric}.  
\begin{construction} [CC reduction from (Gaussian, $\ell_1$) to Hamming cube]\label{cons:embedding} We construct $f$ and $g$ as follows.
\begin{itemize}
    \item [$f$]: Let $n$ be an even number. Given a point $x=(x_1,\dots, x_n)$ sampled from Gaussian space of dimension $n$, do the following:
    \begin{enumerate}
        \item If $\exists i; |x_i| \geq \sqrt{n}/2$, output $0^{n^2}$.
        \item Otherwise, for each $i\in n$ compute $a_i=[\frac{x_i}{\sqrt{n}} + \frac{n}{2}]$ then uniformly sample some $y_i\in\set{0,1}^n$ such that $y_i$ has exactly $a_i$ number of $1$s. Then append $y_i$ s to get $y= (y_1|\dots|y_n)$.
    \end{enumerate}
    \item [$g$]: Let $y=(y_1|\dots| y_n)$ be a Boolean vector of size $n^2$ (each $y_i$ has size $n$). Let $a_i$ be the number of $1$s in $y_i$. Then sample $x= (x_1,\dots, x_n)$ from Gaussian space conditioned on $\frac{2a_i - n}{2\sqrt{n}}\leq x_i < \frac{2a_i - n + 1}{2\sqrt{n}}$
\end{itemize}
\end{construction}

\begin{claim}\label{claim:reduction}
The embedding of Construction \ref{cons:embedding} is an $(\negl(n), 1/\sqrt{n},1/\sqrt{n})$ CC reduction from Gaussian space under $\ell_1$ to Hamming cube (i.e., Boolean hypercube under Hamming distance).
\end{claim}
\begin{proof}
The embedding property of these mappings is proved in \cite{bobkov1997isoperimetric}. The mappings $f$ and $g$ are clearly polynomial time in $n$ and the Almost Lipschitz and Inverse Mappings properties are straightforward.
\end{proof}
The following Corollary follows from Lemma \ref{thm:reduction}, Claim \ref{claim:reduction} and Theorem \ref{thm:main}.
\begin{corollary} [Computational concentration of Gaussian under $\ell_1$]
There is an algorithm $A_{\eps,\delta}^{\cS,\msr}(\cdot)$ that given access to a membership oracle for any set $\cS$ and a sampling oracle from an isotropic Gaussian measure $\msr$ of dimension $n$, it  achieves the following. If $\Pr[\cS] \geq \eps$, given $\eps$ and $\delta$, the  algorithm $A_{\eps,\delta}^{\cS,\msr}(\cdot)$ runs in time $\poly(\nicefrac{n}{\eps \delta})$, and with probability $\ge 1-\delta$ given a random point $\xVec \gets \msr$, it maps $\xVec$ to a point $\yVec \in \cS$ of bounded $\ell_1$ distance $\ell_1(\xVec,\yVec) \leq O(\sqrt{n\cdot \ln(\nicefrac{1}{\eps \delta})})$. 
\end{corollary}

\subsection{Case of Gaussian or Sphere under $\ell_2$}
A reduction may also be used to obtain a (non-optimal) computational concentration of measure for the multi-dimensional Gaussian distribution under the $\ell_2$ metric.

\begin{theorem}
There is an algorithm $A_{\eps,\delta}^{\cS,\msr}(\cdot)$ that given access to a membership oracle for any set $\cS$ and a sampling oracle from an isotropic Gaussian measure $\msr$ of dimension $n$, where each coordinate has variance 1, it  achieves the following. If $\Pr[\cS] \geq \eps$, given $\eps, \delta \ge 1/n^{O(1)}$, the  algorithm $A_{\eps,\delta}^{\cS,\msr}(\cdot)$ runs in time $\poly(n)$, and with probability $\ge 1-\delta$ given a random point $\xVec \gets \msr$, it maps $\xVec$ to a point $\yVec \in \cS$ of bounded $\ell_2$ distance $\ell_2(\xVec,\yVec) \leq O(n^{1/4} \log^{O(1)} n)$. 
\end{theorem}
\begin{proof}
Since $\epsilon \ge 1/n^{O(1)}$, at most $\epsilon/2$ and $\delta/2$ fraction of the points have a coordinate of size $\ge O(\sqrt{\log n})$.
So ignoring points having such large coordinates, we may assume $\Pr[\cS] \ge \epsilon/2$ while every point of $\cS$ has  coordinates as small as $O(\sqrt{\log n})$, and we may assume the point we are mapping also has small coordinates (except our algorithm should now work for $1 - \delta/2$ fraction of the points instead of for $1 - \delta$ fraction.)

Now, when each coordinate is $O(\sqrt{\log n})$,
the $l_2$ distance between two points is at most $O(\sqrt{d_H \log n})$, where $d_H$ is the Hamming distance of the two points.
Now, the theorem follows from our main theorem for Hamming distance.
\end{proof}

We should note that the above computational bound is not information-theoretically tight, since for the Gaussian $\ell_2$ metric probability space, where each coordinate has variance 1, the right bound is $O(\sqrt{\ln(1/(\epsilon \delta))})$.
(This follows e.g.\ from the Gaussian isoperimetric inequality proved in \cite{sudakov1978extremal,borell1975brunn}, which shows the half-space is isopermetrically optimal for the Gaussian distribution.)

Finally, the following shows that our results are not limited to product spaces, and may for example be applied to computational concentration of measure for the high-dimensional sphere.

\begin{theorem}
There is an algorithm $A_{\eps,\delta}^{\cS,\msr}(\cdot)$ that given access to a membership oracle for any set $\cS$ and a sampling oracle from the uniform measure $\msr$ on the unit sphere of dimension $n$, it  achieves the following. If $\Pr[\cS] \geq \eps$, given $\eps, \delta \ge 1/n^{O(1)}$, the  algorithm $A_{\eps,\delta}^{\cS,\msr}(\cdot)$ runs in time $\poly(n)$, and with probability $\ge 1-\delta$ given a random point $\xVec \gets \msr$, it maps $\xVec$ to a point $\yVec \in \cS$ of bounded $\ell_2$ distance $\ell_2(\xVec,\yVec) \leq O(n^{-1/4} \log^{O(1)} n)$. 
\end{theorem}
\begin{proof}
First, we note that a random Gaussian vector, where each coordinate has variance 1, has $\ell_2$ norm $\sqrt{n} + O(n^{1/4})$ except for arbitrary inverse polynomial probability.

So given $\xVec$, we can map it to a new vector $\xVec'$ with the same direction as $\xVec$ but with a random length of distribution square root of chi square, so that the new vector has the Gaussian distribution. We also map the set $\cS$ to the set $\cS' = \{r \cdot s: r \in n^{1/2} + O(n^{1/4}), s \in \cS\}$, where the new set still has probability $\ge \epsilon/2$
under the Gaussian distribution.
By the computational concentration of measure for the Gaussian, we know that we can map, with probability $1-\delta/2$, $\xVec'$ to a point $\yVec' \in \cS'$ of distance $n^{1/4} \log^{O(1)} n$ from $\xVec'$ in $\ell_2$.
Let $\yVec$ be the projection of $\yVec'$ onto the unit sphere.
Therefore 
$$d_{\ell_2}(\xVec, \yVec) \le 
d_{\ell_2}(\xVec, \xVec'/\sqrt{n}) + 
d_{\ell_2}(\xVec'/\sqrt{n}, \yVec'\sqrt{n}) +
d_{\ell_2}(\yVec'/\sqrt{n}, \yVec) = 
O(n^{1/4} \log^{O(1)} n).$$
\end{proof}

These types of relations between concentration of measure of Gaussian and uniform sphere measures has been well-known information-theoretically, e.g.\ see \cite[page 2]{ledoux2001concentration} where concentration for Gaussian is derived from concentration for sphere.
In the above we showed a similar relation for \emph{computational} concentration of measure,
this time deriving for the sphere from the Gaussian.

\section{Computational Concentration around Mean} \label{sec:mean}

 Let $(\cX,\metric,\msr)$ be a metric probability space and $f \colon \cX \To \R$ a measurable function (with respect to $\msr$). For any Borel set $\cT \se \R$, an  parameters $k,\delta \in \R_+$, one can define a computational problem as follows. Given oracle access to a sampler from $\msr$, $\metric$ and function $f(\cdot)$, map a given input $x \in \cX$ algorithmically to $y \in \cY$, such that: (1) $\metric(x,y) \leq k$, and (2) $f(y) \in \cT$ for $1-\delta$ fraction of $x \in \cX$ according to $\msr$. If we already know that $(\cX,\metric,\msr)$ is $(\eps,\delta,k)$ (computationally) concentrates, and if $\Pr_{x \gets \msr}[f(x) \in \cT] \geq \eps$, then it implies that by changing $x$ by at most distance $k$ into a new point $y$, we can (algorithmically) get $f(y) \in \cT$, by defining $\cS = f^{-1}(\cT)$ and noting that $\Pr_\msr[\cS] \geq \eps$. This algorithm needs oracle access to $\cS$
 
\paragraph{Computational concentration around mean.} Again,  let $(\cX,\metric,\msr)$ be a metric probability space and let $f \colon \cX \To \R$ be measurable. Now suppose $\eta = \Ex_{x \gets \msr}[f(x)]$. If we already know, by information theoretic concentration bounds, that $\Pr_{x \gets \msr}[|f(x) - \eta|\leq T]\geq 1-\delta$, then it means that a trivial algorithm that does not even change given $x \gets \msr$, finds a point where $f(x)$  is $T$-close to the average $\eta$. However, this becomes nontrivial, if the goal of the algorithm is to find $y$ that is close to $x$, and that $f(y)$ is much \emph{closer} to the mean $\eta$ than what $x$ achieves. In particular, suppose we somehow know that $\Pr_{x \gets \msr}[|f(x) - \eta|\leq t]\geq \eps$ for $t \ll T, \eps \ll 1-\delta$. (Such results usually follow from the same concentration inequalities proving $\Pr_{x \gets \msr}[|f(x) - \eta|\leq T] \approx 1$.) The smaller $t$ is, the ``higher quality'' the point $x$ has in terms of $f(x)$ being closer to the mean. This means the set $\cS = \set{x \colon |f(x) - \eta|\leq t}$ has  $\msr$ measure at least $\eps$. Therefore, if  the space $(\cX,\metric,\msr)$ is $(\eps,\delta,k)$ computationally concentrated, then we can conclude that there is an efficient algorithm (whose running time can polynomially depend on $1/\eps\delta$ and) that maps $1-\delta$ fraction of $x \gets \msr$ to a point $y \in \cS$. Different, but similar, statements about one-sided concentration can be made as well, if we start from weaker conditions of the form $\Pr_{x \gets \msr}[f(x) >\eta+t]\leq \eps$ (or  $\Pr_{x \gets \msr}[f(x) <\eta-t]\leq 1-\eps$) leading to a weaker conclusion: we can map $x$ to a point $y$ satisfies $f(x) \geq \eta-t$ (or $f(x) \leq \eta+t$). 

Finally, we note that even if the mean $\eta$ is \emph{not} known to the mapping algorithm $A$,  good approximations of it can be obtained by repeated sampling and taking their average. So for simplicity, and without loss of generality, the reader can assume that $\eta$ is known to the mapping algorithm $A$. 

\paragraph{Special case of  Lipschitz functions:  algorithmic proofs of concentration.} When $f \colon \cX \To \R$ is Lipschitz, i.e., $|f(x)-f(y)| \leq \metric(x,y)$, computational concentration around a  set  like $\cS = \set{x \colon |f(x) - \eta|\leq t}$ (or similar one-sided variants) means something stronger than before. We now have an algorithm that \emph{indirectly proves} the concentration around $\eta$ by efficiently finding points that are almost at the border defined by $\eta$. Namely, the Lipschitz now implies that $|f(x) - f(y)| \leq k$, whenever $|x-y|\leq k$. Therefore, the algorithm $A$ mapping $x$ to $y$ is also proving that $1-\delta$ measure of the space $(\cX,\msr)$ is mapped under $f$ to a point that is $k+t$ close to average $\eta$.

All the above arguments are general and apply to any metric probability space. Below, we discuss an special case of a ``McDiarmid type'' inequality in more detail to demonstrate the power of this argument.

\begin{theorem}[An algorithmic variant of McDiarmid inequality] \label{thm:AlgMcD} Suppose $\msr \equiv \msr_1\times \dots \times \msr_n$ is a product measure on a product space $\cX = \cX_1\times \dots \times \cX_n$, and let $f \colon \cX \To \R$ be such that $|f(\xVec)-f(\xVec')|\leq \alpha_i$ whenever $\xVec$ and $\xVec'$ only differ in the $i\th$ coordinate. Let $a=\norm{\alVec}_2$ for $\alVec=(\alpha_1,\dots,\alpha_n)$. Let $\eta = \Ex_{x \gets \msr}[f(\xVec)]$ and
$\cS = \set{x \colon f(\xVec) \leq \eta + \eps \cdot a}$. Then there is an algorithm  $A_{\eps,\delta}^{\msr,f(\xVec)}(\cdot)$ running in time $\poly(\nicefrac{n}{\eps\delta})$ that uses oracle access to $f$ and a sampler from $\msr$, and it holds that   
$$\Pr_{x \gets \msr, y \gets A_{\eps,\delta}^{\msr,f}(\xVec)}\left[\yVec \in  \cS \text{  ~~and~~  } |f(\xVec) - f(\yVec)| \leq O\left(\sqrt{m \cdot \log(\nicefrac{1}{\eps\delta})}\right) \right]\geq 1- \delta.$$

\remove{
Equivalently, if $A_{\eps,k}^{\msr,f(\cdot)}(\cdot)$ is given the ``perturbation budget'' $k$, it holds that 
$$\Pr_{x \gets \msr, y \gets A_{\eps,k}^{\msr,f}(\xVec)}\left[\yVec \in  \cS \text{  ~~and~~  }|f(\xVec) - f(\yVec)| \leq O(k) \right]\geq 1- \frac{e^{-\Omega\left( \nicefrac{k^2}{m}\right)}}{\eps}.$$
\Mnote{actually, did we want to formalize this for an exact given $k$ with the $O(\cdot)$?}
}
\end{theorem}

\paragraph{Corollaries for special cases.} Theorem \ref{thm:AlgMcD} implies a similar result when the quality of the destination region is base on the $\ell_1$ norm; namely,  $\cS = \set{x \colon f(\xVec) \leq \eta + \eps \cdot \norm{\alpha}_1}$, but this follows from the same statement since $\norm{\alpha}_2 \leq \norm{\alpha}_1$. 
In addition, for the special case where $\alpha_i=1$ for all $i$,\footnote{For example, this could be the setting of Hoeffding's inequality in which each coordinate $\msr_i$ is arbitrarily distributed over $[0,1]$, and $f(\xVec)=\sum_{i\in[n]} x_i$, where $\xVec=(x_1,\dots,x_n)$} and let $\gamma,\delta=1/\poly(n)$ be   arbitrarily small inverse polynomials. In that case, Theorem \ref{thm:AlgMcD}, shows that for $1-\delta$ fraction of $\xVec \gets \msr$, we can map $\xVec$ to $\yVec$ in  $\poly(n)$ time in such a way that $f(\yVec) \leq \Ex[f(\msr)]+\gamma$ and $|f(\xVec)-f(\yVec)|\leq \Otilde(\sqrt{n})$. If we choose $\gamma<1/2$, due to the Lipschitz condition,   we can also find some $\yVec$ for which $f(\yVec) \in \Ex[f(\msr)] \pm 1$. This is possible by first finding some $\yVec$ where $f(\yVec) \leq \Ex[f(\msr)]+\gamma$, and then go back over the coordinates in which $\xVec$ and $\yVec$ differ and only changing some of them to get $\yVec'$ where $f(\yVec') \in \Ex[f(\msr)] \pm 1$, and output $\yVec'$ instead. We note that, however, that whenever we want to choose $\gamma<1/2$, we need to also choose $\eps < 1/(2n)$. For this range of small $\eps$, we \emph{cannot} use the computational concentration results of \cite{mahloujifar19-ALT}, but we can indeed use the  stronger computational concentration results of this work that prove computational concentration around any non-negligible event.

\begin{proof}[Proof of Theorem \ref{thm:AlgMcD}]
For starters, suppose $\eta$ is given. In that case, we first observe that $\Pr_\msr[\cS] \geq 1- e^{-2\eps^2} = \Theta(\eps^2)$ by McDiarmid's inequality itself. We can then apply   Theorem \ref{thm:main}. 

When  $\eta$ is not given, we can find a sufficiently good approximation of it, such that $\eta' \in \eta \pm \norm{\alpha}_2 \cdot \eps/10$ (in time $\poly(\nicefrac{n}{\eps\delta})$ and error probability $\delta/10$) and use it instead of $\eta$. Obtaining such $\eta'$ can be done because  any $x,x'$ satisfy $|f(x)-f(x')| \leq \norm{\alpha}_1$. Therefore, we can obtain $\eta' \in \eta \pm \lambda \cdot \norm{\alpha}_1$ in time  by sampling $\ell=\poly(\nicefrac{n}{\lambda\delta})$ (for sufficiently large $\ell$) many points  $\xVec_1,\dots,\xVec_\ell \gets \msr$  and letting $\eta'=\Exp_{i \gets \ell} f(\xVec_i)$. The only catch is that we want $\eta' \in \eta \pm \eps \cdot \norm{\alpha}_2$. However, since it holds that $\norm{\alpha}_2 \leq \norm{\alpha}_1 \cdot \sqrt{n}$, we can choose $\lambda=\eps/\sqrt{n}$, and use the same procedure to obtain $\eta' \in \eta \pm \lambda \cdot \norm{\alpha}_1$ with probability $1-\delta/10$ in time $\poly(\nicefrac{n}{\eps\delta})$.
\end{proof}

\section{Limits of Nonadaptive Methods for Proving Computational Concentration} \label{sec:limits}
\label{sec:lower-bound}
In this section, we consider three restricted types of attacks and prove exponential lower bounds on their running time. The attacks are
\begin{itemize}
    \item {\bf I.i.d.\ queries:} An attack where given $\xVec$, we query i.i.d.\ points whose distribution may depend on $\xVec$, until one of these points lies in $\cS$. The analysis of this attack boils down to analysis of a single-query attack where we want to maximize the probability of $\cS$-membership of the queried point.
    \item {\bf Non-adaptive queries:} An attack where given $\xVec$, we output a list of points, and query all the points in this list. 
    Since the points in the list are determined before the querying, this attack is non-adaptive. It is easy to see (and we give a proof below) how lower bounding this type of attack reduces to the previous type of attack.
    
    \item {\bf Querying only points close enough to have a chance to be output:} If we are interested in finding a point at distance $\le d$ from $\xVec$, one may be tempted to limit the queried points to points at distance $\le d$ from $\xVec$. We show how lower bounding this type of attack reduces to the previous type of attack. 
\end{itemize}

\begin{theorem}[Lower bound for non-adaptive algorithms] \label{thm:nonadaptive-lowerbound}
Let $\msr$ be the uniform probability distribution on $\{1,-1\}^n$, and let $\eps = 1/2$ and $\delta < 1/2$ be constants. There does not exist any non-adaptive algorithm $A$ that given $\xVec\gets \msr$, the algorithm outputs $m = n^{O(1)}$ (random) points $\yVec^1, \ldots, \yVec^m$, all within Hamming distance $n^{1-\Omega(1)}$ of $\xVec$, such that given any set $\cS$ with $\Pr[\cS] \ge \eps$, one of these $m$ points lies in $\cS$ with probability $1-\delta$ over the randomness of $x$ and randomness of $\yVec^1, \ldots, \yVec^m$.
\end{theorem}
\begin{proof} Assume for the sake of contradiction that such an algorithm $A$ exists. Consider the following modified algorithm: given $\xVec$, run $A$ to produce $\yVec^1, \ldots, \yVec^m$, and then let $\zVec^1$ be one of those $m$ vectors uniformly at random.
    To produce $\zVec^2$, run $A$ independently afresh, and let $\zVec^2$ be one of the $m$ freshly produced vectors. We can continue in this way, and produce the vectors $\zVec^1, \ldots, \zVec^{m'}$ as the output of the modified algorithm. By the assumption, for any constant $\delta' \in (\delta, 1/2)$, with probability $1 - \delta'$ over the randomness of $\xVec$, algorithm $A$ has success probability at least $1/n$, hence each $\zVec^i$ lies in $\cS$ with probability $\ge 1/mn$. Hence for these $\xVec$, if we choose $m' = m n^2$, with probability $1-(1-1/mn)^{m'} = 1-o(1)$, the modified algorithm succeeds. Therefore, the average success probability of the algorithm is $\ge 1 - \delta' - o(1) \ge 1/2 + \Omega(1)$.
    
    The above argument shows that we only need to look at algorithms where $\yVec^1, \ldots, \yVec^m$ are independent given $\xVec$. Thus, it is enough to show that there does not exist a random mapping from $\xVec$ to a vector $\yVec$ in such a way that with probability $1 - \delta$ over the randomness of $\xVec$, the probability $\Pr[\yVec \in \cS]$ is non-negligible (since $m$ is polynomial in $n$).
    
    For the sake of contradiction, assume such a mapping from $\xVec$ to $\yVec$ exists. Let $\cS$ be a random half-space, i.e.\ $\cS = \{\zVec: \sum_{i=1}^n a_i z_i \le 0\}$ for a uniformly random vector $a = (a_1, \ldots, a_n) \in \{-1,1\}^n$.
    We will show that for every $\xVec$, with probability $\delta$ over the randomness of $a$, the probability $\Pr[\yVec \in \cS]$ is negligible. By an averaging argument, this shows that there exists a half-space $\cS$ such that with probability $\delta$ over the randomness of $\xVec$, $\Pr[\yVec \in \cS]$ is negligible, completing the proof.
    
    As mentioned above, we want to show that for every $\xVec$, a random half-space is troublesome for the algorithm. By symmetry, without loss of generality, we may assume $\xVec = (1, 1, \ldots, 1)$. Let $\eta = (\eta_1, \ldots, \eta_n) = (\xVec - \yVec)/2$ be the characteristic vector for the coordinates for which $\yVec$ is different from $\xVec$.
    We note that $\yVec \in \cS$ iff $\sum_i a_i - 2\sum_i a_i \eta_i \le 0$.
    We know that with probability $\delta + \Omega(1)$ over the randomness of $a$, we have $\sum_i a_i \ge \Omega(\sqrt{n})$. (This easily follows from the central limit theorem.)
    Now, conditioned on $\eta$, the sum $\sum_i \eta_i a_i$ is actually a sum of $n^{1-\Omega(1)}$-many $\pm 1$ independent random variables of mean zero, so $\Pr[\sum_i \eta_i a_i 
    \ge \Omega(\sqrt{n})]$ is
    a negligible, actually exponentially small, probability.
    This implies over the randomness of $a$ and $\eta$, $\Pr[\sum_i \eta_i a_i \ge \Omega(\sqrt{n})]$ is negligible. Thus, except for an $o(1)$ fraction of random half-spaces, 
    $\Pr[\sum_i \eta_i a_i \ge \Omega(\sqrt{n})]$ is negligible over the randomness of $\yVec$.
    Thus, with probability at least $\delta + \Omega(1) - o(1) \ge \delta$ over the randomness of $a$, we have both
    \begin{itemize}
        \item $\sum_i a_i \ge \Omega(\sqrt{n})$, and
        \item $\Pr[\sum_i a_i \eta_i = \Omega(\sqrt{n})]$  is negligible over the randomness of $\yVec$.
    \end{itemize}
    In this case, $\yVec$ does not lie in $\cS$ except with non-negligible. 
    \end{proof}

\begin{remark}
It can be seen that the above theorem holds whenever $\eps$ and $\delta$ are positive constants such that $\eps + \delta < 1$. 
It can be seen that the above theorem does not hold when $\eps + \delta > 1$ since when we set $\yVec = \xVec$, our failure probability $\delta$ is exactly $1 - \eps$.
\end{remark}

\begin{lemma}
Given a radius $r$, assume an adaptive algorithm $A$, given $\xVec$, wants to find a vector $\yVec \in \cS$ in the ball of radius $r$ around $\xVec$. Furthermore, assume that the algorithm does not make any $\cS$-membership oracle queries regarding points outside the ball. Then, we can transform the algorithm into a non-adaptive algorithm with the same performance.
\end{lemma}
\begin{proof}
When the algorithm ever queries about a point $\yVec$ (and by assumption $\yVec$ is in the ball), if the oracle says that $\yVec \in S$, then we are done (since we have found our desired point.) So the algorithm may always pretend that the result of each membership query about each queried point is that the point is not in $\cS$. This equivalent algorithm is non-adaptive.
\end{proof}

\begin{corollary}
In the $\{0,1\}^n$ uniform product space,
when we want to find a point $\yVec \in \cS$ at distance $n^{1/2+\eps}$ from a random $\xVec$ (for some $\eps \in (0,1/2)$), to be query-efficient, we need to query about $\cS$-membership of points having distance more than $n^{1/2+\eps}$.
\end{corollary}

The above corollary says that even though we are interested in points in a ball of certain radius around $\xVec$, we have to query about points outside that ball. 
When we notice that we are not assuming any structure on the set $\cS$ other than it should have some minimum mass, the above corollary becomes all the more surprising!

\section{Acknowledgement}
We would like to thank Amin Aminzadeh Gohari, Salman Beigi, and Mohammad Hossein Yassaee for helpful discussions.

\bibliographystyle{alpha}
\bibliography{Biblio/references,Biblio/abbrev0,Biblio/crypto,Biblio/CryptoCitations2,Biblio/CryptoCitations}

\newcommand{\etalchar}[1]{$^{#1}$}
\begin{thebibliography}{MHRAR98}

\bibitem[AGK76]{AhlswedeGacsKorner1976}
Rudolf Ahlswede, Peter G{\'a}cs, and J{\'a}nos K{\"o}rner.
\newblock Bounds on conditional probabilities with applications in multi-user
  communication.
\newblock {\em Probability Theory and Related Fields}, 34(2):157--177, 1976.

\bibitem[AM80]{amir1980unconditional}
D~Amir and VD~Milman.
\newblock Unconditional and symmetric sets in n-dimensional normed spaces.
\newblock {\em Israel Journal of Mathematics}, 37(1-2):3--20, 1980.

\bibitem[AM85]{alon1985lambda1}
Noga Alon and Vitali~D Milman.
\newblock $\lambda$1, isoperimetric inequalities for graphs, and
  superconcentrators.
\newblock {\em Journal of Combinatorial Theory, Series B}, 38(1):73--88, 1985.

\bibitem[B{\etalchar{+}}97]{bobkov1997isoperimetric}
Sergey~G Bobkov et~al.
\newblock An isoperimetric inequality on the discrete cube, and an elementary
  proof of the isoperimetric inequality in gauss space.
\newblock {\em The Annals of Probability}, 25(1):206--214, 1997.

\bibitem[BEG17]{BeigiEG17}
Salman Beigi, Omid Etesami, and Amin Gohari.
\newblock Deterministic randomness extraction from generalized and distributed
  santha--vazirani sources.
\newblock {\em SIAM Journal on Computing}, 46(1):1--36, 2017.

\bibitem[BEK02]{NastyNoise}
Nader~H. Bshouty, Nadav Eiron, and Eyal Kushilevitz.
\newblock {{PAC} learning with nasty noise}.
\newblock {\em Theoretical Computer Science}, 288(2):255--275, 2002.

\bibitem[BFR14]{biggio2014security}
Battista Biggio, Giorgio Fumera, and Fabio Roli.
\newblock Security evaluation of pattern classifiers under attack.
\newblock {\em IEEE transactions on knowledge and data engineering},
  26(4):984--996, 2014.

\bibitem[BGZ16]{bentov2016bitcoin}
Iddo Bentov, Ariel Gabizon, and David Zuckerman.
\newblock Bitcoin beacon.
\newblock {\em arXiv preprint arXiv:1605.04559}, 2016.

\bibitem[BHT14]{berman2014coin}
Itay Berman, Iftach Haitner, and Aris Tentes.
\newblock Coin flipping of any constant bias implies one-way functions.
\newblock In {\em Proceedings of the 46th Annual ACM Symposium on Theory of
  Computing}, pages 398--407. ACM, 2014.

\bibitem[BNL12]{biggio2012poisoning}
Battista Biggio, Blaine Nelson, and Pavel Laskov.
\newblock Poisoning attacks against support vector machines.
\newblock In {\em Proceedings of the 29th International Coference on
  International Conference on Machine Learning}, pages 1467--1474. Omnipress,
  2012.

\bibitem[BNS{\etalchar{+}}06]{barreno2006can}
Marco Barreno, Blaine Nelson, Russell Sears, Anthony~D Joseph, and J~Doug
  Tygar.
\newblock Can machine learning be secure?
\newblock In {\em Proceedings of the 2006 ACM Symposium on Information,
  computer and communications security}, pages 16--25. ACM, 2006.

\bibitem[BOL89]{ben1989collective}
M.~Ben-Or and N.~Linial.
\newblock Collective coin flipping.
\newblock {\em Randomness and Computation}, 5:91--115, 1989.

\bibitem[Bor75]{borell1975brunn}
Christer Borell.
\newblock The brunn-minkowski inequality in gauss space.
\newblock {\em Inventiones mathematicae}, 30(2):207--216, 1975.

\bibitem[BPR18]{bubeck2018adversarial}
S{\'e}bastien Bubeck, Eric Price, and Ilya Razenshteyn.
\newblock {Adversarial examples from computational constraints}.
\newblock {\em arXiv preprint arXiv:1805.10204}, 2018.

\bibitem[CG88]{chor1988unbiased}
Benny Chor and Oded Goldreich.
\newblock Unbiased bits from sources of weak randomness and probabilistic
  communication complexity.
\newblock {\em SIAM Journal on Computing}, 17(2):230--261, 1988.

\bibitem[CI93]{cleve1993martingales}
Richard Cleve and Russell Impagliazzo.
\newblock Martingales, collective coin flipping and discrete control processes.
\newblock {\em Manuscript}, 1993.

\bibitem[CW17]{CarliniWagner}
Nicholas Carlini and David~A. Wagner.
\newblock {Towards Evaluating the Robustness of Neural Networks}.
\newblock In {\em 2017 {IEEE} Symposium on Security and Privacy, {SP} 2017, San
  Jose, CA, USA, May 22-26, 2017}, pages 39--57, 2017.

\bibitem[DKK{\etalchar{+}}16]{diakonikolas2016robust}
Ilias Diakonikolas, Gautam Kamath, Daniel~M Kane, Jerry Li, Ankur Moitra, and
  Alistair Stewart.
\newblock Robust estimators in high dimensions without the computational
  intractability.
\newblock In {\em Foundations of Computer Science (FOCS), 2016 IEEE 57th Annual
  Symposium on}, pages 655--664. IEEE, 2016.

\bibitem[DKK{\etalchar{+}}18]{diakonikolas2018sever}
Ilias Diakonikolas, Gautam Kamath, Daniel~M Kane, Jerry Li, Jacob Steinhardt,
  and Alistair Stewart.
\newblock Sever: A robust meta-algorithm for stochastic optimization.
\newblock {\em arXiv preprint arXiv:1803.02815}, 2018.

\bibitem[DOPS04]{DodisOnPrSa04}
Yevgeniy Dodis, Shien~Jin Ong, Manoj Prabhakaran, and Amit Sahai.
\newblock {On the (Im)possibility of Cryptography with Imperfect Randomness}.
\newblock In {\em FOCS: IEEE Symposium on Foundations of Computer Science
  (FOCS)}, 2004.

\bibitem[DV19]{degwekar2019computational}
Akshay Degwekar and Vinod Vaikuntanathan.
\newblock Computational limitations in robust classification and win-win
  results.
\newblock {\em arXiv preprint arXiv:1902.01086}, 2019.

\bibitem[GKP15]{goldwasser2015adaptively}
Shafi Goldwasser, Yael~Tauman Kalai, and Sunoo Park.
\newblock Adaptively secure coin-flipping, revisited.
\newblock In {\em International Colloquium on Automata, Languages, and
  Programming}, pages 663--674. Springer, 2015.

\bibitem[GMP18]{GoodfellowEtAl:MakeMLRobust}
Ian~J. Goodfellow, Patrick~D. McDaniel, and Nicolas Papernot.
\newblock {Making machine learning robust against adversarial inputs}.
\newblock {\em Communications of the {ACM}}, 61(7):56--66, 2018.

\bibitem[HO14]{haitner2014coin}
Iftach Haitner and Eran Omri.
\newblock Coin flipping with constant bias implies one-way functions.
\newblock {\em SIAM Journal on Computing}, 43(2):389--409, 2014.

\bibitem[IK10]{impagliazzo2010constructive}
Russell Impagliazzo and Valentine Kabanets.
\newblock Constructive proofs of concentration bounds.
\newblock In {\em Approximation, Randomization, and Combinatorial Optimization.
  Algorithms and Techniques}, pages 617--631. Springer, 2010.

\bibitem[KKR18]{kalai2018lower}
Yael~Tauman Kalai, Ilan Komargodski, and Ran Raz.
\newblock A lower bound for adaptively-secure collective coin-flipping
  protocols.
\newblock In {\em 32nd International Symposium on Distributed Computing}, 2018.

\bibitem[KL93]{KearnsLi::Malicious}
Michael~J. Kearns and Ming Li.
\newblock {Learning in the Presence of Malicious Errors}.
\newblock {\em {SIAM} Journal on Computing}, 22(4):807--837, 1993.

\bibitem[Led01]{ledoux2001concentration}
Michel Ledoux.
\newblock {\em {The Concentration of Measure Phenomenon}}.
\newblock Number~89 in Mathematical Surveys and Monographs. American
  Mathematical Society, 2001.

\bibitem[L{\'e}v51]{levy1951problemes}
Paul L{\'e}vy.
\newblock {\em Probl{\`e}mes concrets d'analyse fonctionnelle}, volume~6.
\newblock Gauthier-Villars Paris, 1951.

\bibitem[LLS89]{lichtenstein1989some}
David Lichtenstein, Nathan Linial, and Michael Saks.
\newblock Some extremal problems arising from discrete control processes.
\newblock {\em Combinatorica}, 9(3):269--287, 1989.

\bibitem[LRV16]{lai2016agnostic}
Kevin~A Lai, Anup~B Rao, and Santosh Vempala.
\newblock Agnostic estimation of mean and covariance.
\newblock In {\em Foundations of Computer Science (FOCS), 2016 IEEE 57th Annual
  Symposium on}, pages 665--674. IEEE, 2016.

\bibitem[Mar74]{Margulis1974}
Grigorii~Aleksandrovich Margulis.
\newblock Probabilistic characteristics of graphs with large connectivity.
\newblock {\em Problemy peredachi informatsii}, 10(2):101--108, 1974.

\bibitem[Mar86]{Marton1986}
Katalin Marton.
\newblock A simple proof of the blowing-up lemma (corresp.).
\newblock {\em IEEE Transactions on Information Theory}, 32(3):445--446, 1986.

\bibitem[McD89]{mcdiarmid1989method}
Colin McDiarmid.
\newblock On the method of bounded differences.
\newblock {\em Surveys in combinatorics}, 141(1):148--188, 1989.

\bibitem[MDM18]{Mahloujifar2018:ALT}
Saeed Mahloujifar, Dimitrios~I Diochnos, and Mohammad Mahmoody.
\newblock {Learning under $p$-Tampering Attacks}.
\newblock In {\em ALT}, pages 572--596, 2018.

\bibitem[MDM19]{mahloujifar2018curse}
Saeed Mahloujifar, Dimitrios~I Diochnos, and Mohammad Mahmoody.
\newblock The curse of concentration in robust learning: Evasion and poisoning
  attacks from concentration of measure.
\newblock {\em AAAI Conference on Artificial Intelligence}, 2019.

\bibitem[MFF16]{DeepFool}
Seyed{-}Mohsen Moosavi{-}Dezfooli, Alhussein Fawzi, and Pascal Frossard.
\newblock {DeepFool: A Simple and Accurate Method to Fool Deep Neural
  Networks}.
\newblock In {\em CVPR}, pages 2574--2582, 2016.

\bibitem[MHRAR98]{mcdiarmid1998probabilistic}
Colin Mcdiarmid, M~Habib, J~Ramirez-Alfonsin, and B~Reed.
\newblock Probabilistic methods for algorithmic discrete mathematics.
\newblock {\em Algorithms and Combinatorics Series}, 16:1--46, 1998.

\bibitem[MM17]{pTampTCC17}
Saeed Mahloujifar and Mohammad Mahmoody.
\newblock Blockwise $p$-tampering attacks on cryptographic primitives,
  extractors, and learners.
\newblock In {\em Theory of Cryptography Conference}, pages 245--279. Springer,
  2017.

\bibitem[MM19]{mahloujifar19-ALT}
Saeed Mahloujifar and Mohammad Mahmoody.
\newblock Can adversarially robust learning leverage computational hardness?
\newblock In Aur\'elien Garivier and Satyen Kale, editors, {\em Proceedings of
  the 30th International Conference on Algorithmic Learning Theory}, volume~98
  of {\em Proceedings of Machine Learning Research}, pages 581--609, Chicago,
  Illinois, 22--24 Mar 2019. PMLR.

\bibitem[Mos09]{moser2009constructive}
Robin~A Moser.
\newblock A constructive proof of the lov{\'a}sz local lemma.
\newblock In {\em Proceedings of the forty-first annual ACM symposium on Theory
  of computing}, pages 343--350. ACM, 2009.

\bibitem[MPS10]{maji2010computational}
Hemanta~K Maji, Manoj Prabhakaran, and Amit Sahai.
\newblock On the computational complexity of coin flipping.
\newblock In {\em Foundations of Computer Science (FOCS), 2010 51st Annual IEEE
  Symposium on}, pages 613--622. IEEE, 2010.

\bibitem[MS86]{milman1986asymptotic}
Vitali~D Milman and Gideon Schechtman.
\newblock {\em Asymptotic theory of finite dimensional normed spaces}, volume
  1200.
\newblock Springer Verlag, 1986.

\bibitem[MT10]{moser2010constructive}
Robin~A Moser and G{\'a}bor Tardos.
\newblock A constructive proof of the general lov{\'a}sz local lemma.
\newblock {\em Journal of the ACM (JACM)}, 57(2):11, 2010.

\bibitem[PSBR18]{prasad2018robust}
Adarsh Prasad, Arun~Sai Suggala, Sivaraman Balakrishnan, and Pradeep Ravikumar.
\newblock Robust estimation via robust gradient estimation.
\newblock {\em arXiv preprint arXiv:1802.06485}, 2018.

\bibitem[RVW04]{ReingoldVW04}
Omer Reingold, Salil Vadhan, and Avi Wigderson.
\newblock A note on extracting randomness from santha-vazirani sources.
\newblock {\em Unpublished manuscript}, 2004.

\bibitem[ST78]{sudakov1978extremal}
Vladimir~N Sudakov and Boris~S Tsirel'son.
\newblock Extremal properties of half-spaces for spherically invariant
  measures.
\newblock {\em Journal of Mathematical Sciences}, 9(1):9--18, 1978.

\bibitem[SV86]{SanthaV86}
Miklos Santha and Umesh~V. Vazirani.
\newblock Generating quasi-random sequences from semi-random sources.
\newblock {\em J. Comput. Syst. Sci.}, 33(1):75--87, 1986.

\bibitem[SZS{\etalchar{+}}14]{Szegedy:intriguing}
Christian Szegedy, Wojciech Zaremba, Ilya Sutskever, Joan Bruna, Dumitru Erhan,
  Ian Goodfellow, and Rob Fergus.
\newblock Intriguing properties of neural networks.
\newblock In {\em ICLR}, 2014.

\bibitem[Tal95]{talagrand1995concentration}
Michel Talagrand.
\newblock Concentration of measure and isoperimetric inequalities in product
  spaces.
\newblock {\em Publications Math{\'e}matiques de l'Institut des Hautes Etudes
  Scientifiques}, 81(1):73--205, 1995.

\bibitem[Val85]{Valiant::DisjunctionsConjunctions}
Leslie~G. Valiant.
\newblock {Learning disjunctions of conjunctions}.
\newblock In {\em IJCAI}, pages 560--566, 1985.

\end{thebibliography}

\end{document}